\documentclass[a4paper,11pt]{quantumarticle}
\pdfoutput=1
\usepackage{float}
\usepackage{graphicx} 
\usepackage{amssymb}
\usepackage[utf8]{inputenc}
\usepackage[T1]{fontenc}
\usepackage{amsfonts}
\usepackage{bm}
\usepackage{physics}
\usepackage{amsmath}
\usepackage{amsthm}
\usepackage{mathtools}
\usepackage[numbers]{natbib}
\usepackage{stmaryrd}
\usepackage{tikz}
\usepackage{quantikz}
\usepackage{hyperref}
\usepackage{makecell}
\usetikzlibrary{decorations.pathreplacing}
\def\id{{\rm 1\kern-.22em l}}

\newcommand{\stirling}[2]{\genfrac{[}{]}{0pt}{}{#1}{#2}}
\newtheorem{tw}{Theorem}
\newtheorem{propo}[tw]{Proposition}
\newtheorem{lem}[tw]{Lemma}
\newtheorem{df}[tw]{Definition}
\newtheorem{uw}[tw]{Remark}
\newtheorem{prz}[tw]{Example}
\newtheorem{cor}[tw]{Corollary}

\begin{document}

\title{Asymptotic distinguishability of Haar-averaged measurement models}
\author{Ludmiła Marcinkowska}
\affiliation{Institute of Theoretical and Applied Informatics, Polish Academy of Sciences, ul. Ba{\l}tycka 5, 44-100 Gliwice, Poland}
\orcid{0009-0002-0935-3921}
\author{Łukasz Pawela}
\affiliation{Institute of Theoretical and Applied Informatics, Polish Academy of Sciences, ul. Ba{\l}tycka 5, 44-100 Gliwice, Poland}
\orcid{0000-0002-0476-7132}
\author{Marcin Markiewicz}
\affiliation{Institute of Theoretical and Applied Informatics, Polish Academy of Sciences, ul. Ba{\l}tycka 5, 44-100 Gliwice, Poland}
\affiliation{International Centre for Theory of Quantum Technologies, University of Gdansk, 80-309 Gda{\'n}sk, Poland}
\orcid{0000-0002-8983-9077}
\author{Zbigniew Puchała}
\affiliation{Institute of Theoretical and Applied Informatics, Polish Academy of Sciences, ul. Ba{\l}tycka 5, 44-100 Gliwice, Poland}
\orcid{0000-0002-4739-0400}

\begin{abstract} We study discrimination problems generated by the same basic
Haar-random measurement mechanism at two observational levels. First, we derive
an explicit expression for the type-II error in the task of discriminating a
Haar-random measure-and-prepare channel from the identity channel
$\mathbb{I}$, using a coherence-sensitive entangled tester. Second, after
passing to the induced classical measurement records, we compare two random
measurement models: one induced by a single collective unitary of the form
$U^{\otimes (n_1+n_2)}$ with $U\in U(d)$, and another induced by independent
local unitaries $U_1^{\otimes n_1}\otimes U_2^{\otimes n_2}$. For the
associated Haar-averaged aggregate histogram laws, in which the block of origin
of each count is not retained, we obtain closed-form formulas and quantify
their discrepancy through the total variation distance. We derive asymptotic
expressions in the fixed-$N$, large-$d$ regime, the fixed-$d$, large-$N$ regime,
the sparse joint-scaling regime $N=o(\sqrt d)$, and the critical scaling regime
$N/\sqrt d\to c$. We also identify the block-resolved pair-of-histograms law,
showing that the aggregate total variation distance is a coarse-grained lower
bound on the distinguishability available when block labels are retained.\end{abstract}

\maketitle
\onecolumn

\section{Introduction and motivation}
Distinguishing quantum processes is a central problem in quantum information
theory, with key applications in device verification, quantum communication, and
process tomography. Formally, it can be cast as a task of quantum hypothesis
testing, in which a decision must be made between two candidate quantum channels
based on a limited number of channel uses and the resulting measurement outcomes
(see
\cite{chefles2000quantum},\cite{barnett2009quantum},\cite{helstrom1969quantum},\cite{holevo1973statistical},\cite{Hayashi09},\cite{Zhuang20},\cite{Watrous18}).
A particularly well-studied special case concerns the discrimination of unitary channels, where the structural properties of unitary operations --- such as the role of entanglement, the feasibility of perfect discrimination, and the use of adaptive or multi-shot strategies --- have been analyzed in detail (see \cite{Duan07},\cite{Duan09},\cite{Hillery10},\cite{Soeda21},\cite{Hashimoto22}).
The discrimination of quantum measurements, which is particularly relevant in
experimentally accessible scenarios, has also been widely investigated,
including single-shot and multi-shot settings as well as unambiguous
discrimination strategies (see
\cite{puchala2018strategies},\cite{puchala2021multiple},\cite{krawiec2024discrimination},\cite{ziman2009unambiguous},\cite{sedlak2014optimal}).

A powerful approach to the study of typical properties of quantum channels
relies on averaging over the Haar measure on the unitary group. In this setting,
the Weingarten calculus provides explicit formulas for integrals of polynomial
functions of unitary matrices and has become a standard tool in the analysis of
random measurement operators, random quantum channels, and quantum networks (see
\cite{puchala2011symbolic},\cite{collins2006integration},\cite{chiribella2009theoretical}). Random quantum
channels generated from Haar-random constructions and random POVMs induced by
random isometries have been extensively studied in recent years, see
\cite{heinosaari2020random},\cite{kukulski2021generating},\cite{nechita2018almost}.

Rather than characterizing channels through their full Choi representations, we
adopt a more operational perspective that focuses directly on the statistics of
measurement outcomes they induce. This approach emphasizes experimentally
accessible data and avoids the overhead associated with full process tomography.
It also aligns naturally with practical verification in which only measurement
statistics are available.

The two discrimination problems studied in this paper are organized around the
same basic object: a Haar-random unitary followed by a projective measurement,
or equivalently a Haar-random measure-and-prepare channel. The first problem
uses a coherence-sensitive quantum tester and asks how well this channel family
can be distinguished from the identity channel $\mathbb{I}$. This provides a
baseline at the channel level, before any restriction to classical records is
imposed.

The second problem keeps the same random-measurement mechanism but changes the
available data. Instead of testing the channel against the identity, we look at
the classical outcome statistics produced after the measurement and ask whether
these statistics reveal the structure of the underlying Haar randomness. More
precisely, we compare a model generated by a single Haar-random unitary acting
collectively (as a tensor power) on a system composed of $N=n_1 + n_2$
subsystems with a model generated by two independent Haar-random unitaries
acting collectively on $n_1$ and $n_2$ subsystems. While both models produce
exchangeable outcome statistics, their Haar-averaged histogram distributions
differ in a nontrivial combinatorial manner.

A useful way to understand the statistical structure of random quantum
measurements is through an analogy with classical occupancy processes. When a
fixed input state is measured repeatedly in a basis obtained from a Haar-random
unitary, each outcome can be viewed as drawing a symbol from a probability
vector given by the squared amplitudes of a random column of the unitary matrix.

Conditioned on the unitary, the measurement outcomes are independent and follow
a multinomial distribution. Averaging over the Haar measure therefore induces a
classical random process in which the probability vector itself is random and
uniformly distributed on the probability simplex. As a consequence, the
resulting statistics can be interpreted as a classical occupancy problem:
distributing $n$ balls (measurement shots) into $d$ bins (basis outcomes).

In this picture, the most informative events are collisions, i.e. repeated
occurrences of the same outcome. When the number of samples is small compared to
the Hilbert-space dimension, collisions are rare and the outcome statistics are
close to those generated by independent sampling from a large alphabet. In
contrast, when the number of samples grows, collisions become more frequent and
reveal correlations imposed by the underlying measurement model.

This observation provides a natural mechanism for distinguishing different
random measurement structures. In particular, correlations between subsystems in
a collective Haar-random unitary influence the probability of collision events
across blocks of measurement outcomes. By contrast, when the measurement is
generated by independent unitaries acting collectively on separate subsystems,
such cross-block correlations are suppressed. As we show below, the
distinguishability between these two models can therefore be traced directly to
differences in their collision statistics.

Thus the first part fixes the channel-level benchmark, while the main technical
part of the paper analyzes what remains visible after one passes to classical
measurement statistics. In that setting, the central quantity is the total
variation distance (TVD) between the Haar-averaged aggregate histogram laws. The
TVD provides an operationally meaningful measure of statistical
distinguishability for the block-unresolved experiment in which only the total
histogram is observed. We also record the corresponding block-resolved law for
the ordered pair of histograms, which gives the operational distinguishability
when the block labels are retained. Our analysis shows that the
distinguishability between collective and block-collective random measurements
is governed by collision statistics, leading to four distinct asymptotic regimes
depending on the scaling of the number of samples $N$ and the Hilbert-space
dimension $d$.

\subsection{Contributions}
In this work we provide a detailed analysis of discrimination problems for
Haar-random quantum measurement models, moving from a channel-level test to the
classical statistics induced by the same random measurement construction. Our
main contributions are as follows:
\begin{enumerate}
    \item We derive an explicit formula for the Haar-averaged type-II error
    probability in the problem of distinguishing a Haar-random
    measure-and-prepare channel from the identity channel using a maximally
    entangled input and a two-outcome tester. This establishes a channel-level
    benchmark for the random measurement family, expressed as a unitary integral
    evaluated via Weingarten calculus in Theorem~\ref{typeII}.
    \item We show that Haar-averaged measurement statistics can be mapped to a
    classical occupancy problem, where outcome histograms follow a
    Dirichlet-induced distribution on the probability simplex. This provides a
    transparent probabilistic interpretation of quantum measurement statistics
    in terms of collisions.
    \item For both the collective and block-collective random measurement
    models, we derive closed-form expressions for the Haar-averaged aggregate
    histogram laws, and show that the block model induces a nontrivial
    combinatorial deformation governed by constrained histogram decompositions.
    \item We express the aggregate total variation distance $\mathrm{TVD}$
    between the two models in terms of a combinatorial functional counting
    admissible histogram decompositions. This identifies collision events as the
    fundamental mechanism responsible for statistical distinguishability.
    \item We derive the exact block-resolved pair-of-histograms law. This
    clarifies the operational meaning of the aggregate statistic: the
    total-histogram TVD is a coarse-grained lower bound on the TVD obtainable
    from block-resolved data, and in the fixed-$d$, large-$N$ regime the
    block-resolved TVD tends to one.
    \item We analyze the behavior of the aggregate TVD in several asymptotic
    regimes:
        \begin{itemize}
            \item for fixed $N$ and $d \to \infty$, we show that the
            distinguishability is governed by rare collision events and vanishes
            with the dimension,
            \item for fixed $d$ and $N \to \infty$, we derive a simplex-limit
            formula expressed as an explicit integral over the probability
            simplex,
            \item in the sparse joint-scaling regime $N=o(\sqrt{d})$, we recover
            the same collision-driven asymptotic behaviour,
            \item in the critical scaling regime $N/\sqrt d\to c$, we prove that
            the collision count converges to a Poisson random variable and
            derive an exact expression for the limiting aggregate $\mathrm{TVD}$
            as an expectation over this distribution (Theorem~\ref{twCritScal}),
            providing a complete description of the transition between sparse
            and dense sampling regimes.
        \end{itemize}
\end{enumerate}

\section{\texorpdfstring{Distinguishing a measurement channel from the identity
channel $\mathbb{I}$}{Distinguishing a measurement channel from the identity
channel I }}

In this section we analyze the problem of distinguishing the identity channel
$\mathbb{I}$ from a random quantum channel implemented as a measurement in a
collectively rotated computational basis, using a fixed two-outcome tester and a
maximally entangled input state. Our goal is to compute the Haar-averaged
probability of a type II error and to express it in terms of a Haar integral
over the unitary group. For each $U \in U(d)$ we define a measure-and-prepare
channel $Q^{(N)}_U$ acting on $(\mathbb{C}^d)^{\otimes N}$, defined as follows.
Let $U\in U(d)$ be a Haar-random unitary and let $\{\ket{k}\}_{k=0}^{d-1}$
denote the computational basis of $\mathbb{C}^d$. For a multi-index
$\mathbf{k}=(k_0,\dots,k_{N-1})\in[d]^N$, define, in the Heisenberg picture,
\begin{equation}
    \begin{split}
        & E_{U,k}:=U^\dagger\ket{k}\!\bra{k}U,\\
        & E_{U,\mathbf{k}} := E_{U,k_0}\otimes\cdots\otimes E_{U,k_{N-1}}.
    \end{split}
\end{equation}
The channel $Q^{(N)}_U$ is then given by
\begin{equation} \label{defQ}
\begin{split}
    &Q^{(N)}_U:\rho \mapsto \sum_{\mathbf{k}\in[d]^N} \big(\mathrm{Tr}\, E_{U,\mathbf{k}}\,\rho\big)\,\ket{\mathbf{k}}\!\bra{\mathbf{k}},\\
    \end{split}
\end{equation}
where $\{\ket{\mathbf{k}}\}$ denotes the standard product basis of
$(\mathbb{C}^d)^{\otimes N}$. Equivalently
\begin{equation}
     Q^{(N)}_U(\rho)= \Delta\left( U^{\otimes N}\,\rho\,U^{\dagger\otimes N} \right),
\end{equation}
where
\begin{equation}
    \Delta(\rho):=\sum_{k} \bra{k}\rho\ket{k}\ket{k}\bra{k}.
\end{equation}
This channel corresponds to a measurement in the rotated computational basis
followed by a classical encoding of the outcome. 

We consider the hypothesis 
\begin{equation}
    \begin{split}
        & H_0: X=\mathbb{I}, \\
        & H_1: X=Q^{(N)}_U\neq \mathbb{I},
    \end{split}
\end{equation}
and a bipartite system $\mathcal{H}_A\otimes\mathcal{H}_B \cong
(\mathbb{C}^d)^{\otimes N}\otimes(\mathbb{C}^d)^{\otimes N}$. The unknown
channel $X\in\{\mathbb{I},Q^{(N)}_U\}$ acts on subsystem $\mathcal{H}_A$, while
the identity channel acts on $\mathcal{H}_B$. The discrimination strategy
consists of preparing a maximally entangled state between $\mathcal{H}_A$ and
$\mathcal{H}_B$,
$\ket{\psi}=\frac{1}{\sqrt{d^N}}\sum_{\mathbf{k}}\ket{\mathbf{k}\mathbf{k}}$,
followed by a joint measurement on the output state.

The two-outcome measurement $\{\Pi_0, \Pi_1\}$
is performed jointly on $\mathcal{H}_A\otimes\mathcal{H}_B$ and is defined as 
\begin{equation} \label{tester}
    \begin{split}
        & \Pi_0=\frac{1}{d^N}\sum_{\mathbf{k},\mathbf{l}} \ket{\mathbf{k}\mathbf{k}}\!\bra{\mathbf{l}\mathbf{l}}= \proj{\psi}, \\
        & \Pi_1=\mathbf{1}-\Pi_0.
    \end{split}
\end{equation}
Let us denote the resulting output states by
\begin{equation} \label{outputs}
     \rho_0 := \proj{\psi}, \qquad \rho_U:= (Q^{(N)}_U \otimes \mathbb{I})(\proj{\psi}).
\end{equation}
The measurement outcome corresponding to $\Pi_0$ is interpreted as accepting
$H_0$, while the outcome corresponding to $\Pi_1$ is interpreted as accepting
$H_1$. The type-I and type-II errors are then given by
\begin{equation}
\begin{split}
    &  p_{I}=\mathrm{Tr}(\rho_0\Pi_1),\\
    &  p_{II}=\mathrm{Tr}(\rho_U \Pi_0),\\
\end{split}
\end{equation} 
where the type-I error corresponds to incorrectly identifying the channel as
different from the identity when it is in fact the identity, and the type-II
error corresponds to failing to detect a difference when the channel is
$Q^{(N)}_U$. 

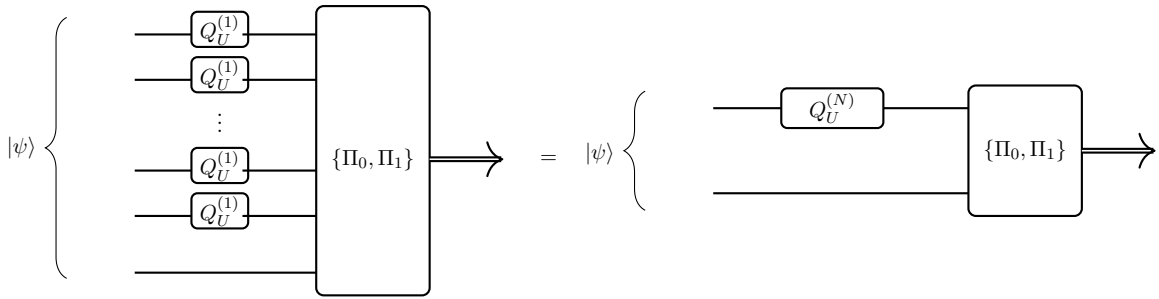
\begin{figure}[H]
\centering
\begin{tikzpicture}[scale=0.75, transform shape]
\draw[decorate,decoration={brace,amplitude=8pt,mirror}]
(-1.2,3.5) -- (-1.2,-1.1);
\node[left] at (-1.6,1.2)
{$|\psi\rangle$};
\foreach \y in {3.2,2.4,0.8,0}
{
    \draw[thick] (0,\y) -- (1,\y);
    \draw[thick, rounded corners=0.08cm]
    (1,\y-0.22) rectangle +(1.0,0.62);
    \node at (1.5,\y+0.08) {$Q^{(1)}_U$};
    \draw[thick] (1.9,\y) -- (3.2,\y);
}
\node at (1.5,1.75) {$\vdots$};
\draw[thick] (0,-1) -- (3.2,-1);
\draw[thick, rounded corners=0.1cm]
(3.2,-1.4) rectangle (5.2,3.7);
\node at (4.2,1)
{$\{\Pi_0,\Pi_1\}$};
\draw[double, thick, ->]
(5.2,1) -- (6.5,1);
\node at (7.3,1) {$=$};
\draw[decorate,decoration={brace,amplitude=8pt,mirror}]
(9,2.2) -- (9,0.1);
\node[left] at (8.6,1.1)
{$|\psi\rangle$};
\draw[thick] (10.2,1.9) -- (11.4,1.9);
\draw[thick, rounded corners=0.08cm]
(11.4,1.55) rectangle +(1.8,0.7);
\node at (12.3,1.9)
{$Q_U^{(N)}$};
\draw[thick] (13.2,1.9) -- (14.7,1.9);
\draw[thick] (10.2,0.4) -- (14.7,0.4);
\draw[thick, rounded corners=0.1cm]
(14.7,0.0) rectangle (16.7,2.3);
\node at (15.7,1.15)
{$\{\Pi_0,\Pi_1\}$};
\draw[double, thick, ->]
(16.7,1.15) -- (18,1.15);
\end{tikzpicture}

\caption{Discrimination protocol between the identity channel $\mathbb{I}$ and
the Haar-random measurement channel $Q_U^{(N)}$. The maximally entangled state
$\ket{\psi}$ is prepared across $N$ subsystems. The channel $Q_U^{(N)}$ acts
collectively on these subsystems, followed by the joint measurement
$\{\Pi_0,\Pi_1\}$.}
\end{figure}

\begin{tw}\label{typeII}
For the above discrimination strategy the following holds:
\begin{enumerate}
\item If $X=\mathbb{I}$, the outcome corresponding to $\Pi_0$ is obtained with
probability $1$.
\item If $X=Q^{(N)}_U$, then the Haar-averaged type–II error
\begin{equation} 
\overline{p}_{II}:=\int p_{II}(U)dU
\end{equation}
is given by
\begin{equation} \label{defInd}
    \overline{p}_{II} =\frac{1}{d^{2N}}\sum_{\sigma,\tau\in S_N} d^{c(\sigma\vee\tau)} \mathrm{Wg}(\sigma^{-1}\tau,d)
\end{equation}
where $\mathrm{Wg}$ denotes the Weingarten function and $c(\sigma\vee \tau)$
denotes the number of blocks in the join of the set partitions induced by the
cycle decompositions of $\sigma$ and $\tau$.
\end{enumerate}
\end{tw}
\begin{proof}
The first statement follows from the definition of $\Pi_0$ and $\ket{\psi}$;
\begin{equation}
    \begin{split}
        & \text{Tr}\proj{\psi}\Pi_0= 1.
    \end{split}
\end{equation}
For $X=Q^{(N)}_U$, define the Haar integral
\begin{equation}
I_{N,d}:=\sum_{\mathbf{k}}\int\left|\bra{\mathbf{k}}U^{\otimes N}\ket{\mathbf{k}} \right|^2 dU .
\end{equation}
Then
\begin{equation}
\overline{p}_{II}=\frac{1}{d^{2N}} I_{N,d}.
\end{equation}
A straightforward computation using
\eqref{defQ}-\eqref{outputs} shows that
\begin{equation}
I_{N,d}=\sum_{\mathbf{k}} \int \bra{\mathbf{k}}E_{U,\mathbf{k}}\ket{\mathbf{k}} dU= \int\Big(\sum_{j=0}^{d-1} |U_{jj}|^2\Big)^N dU.
\end{equation}
Using Weingarten calculus, we obtain
\begin{equation}
I_{N,d}=\sum_{\sigma,\tau\in S_N}d^{c(\sigma\vee\tau)}\mathrm{Wg}(\sigma^{-1}\tau,d).
\end{equation}
The reduction to the Haar integral is purely algebraic; details are provided in Appendix~\ref{ObInd}.
\end{proof}
\begin{propo}[Large-$d$ type-II asymptotics]\label{typeIIasymp}
    Fix $N\ge1$. As $d \to \infty$, with constants in the $O(\cdot)$ terms
    depending only on $N$, the following asymptotic expansion holds:
   \begin{equation}
    I_{N,d} = 1+\frac{\binom{N}{2}}{d} +O(d^{-2}),
    \end{equation}
    \begin{equation}
    \overline{p}_{II}=\frac{1}{d^{2N}}\left(1+\frac{\binom{N}{2}}{d} +O(d^{-2})\right) \xrightarrow[]{d \to \infty}0.
    \end{equation}
    Hence, for fixed $N$, the Haar-averaged type-II error vanishes
    asymptotically. 
\end{propo}
The proof is based on the Weingarten expansion and an analysis of the
contributions of permutations in $S_N$. It is presented in
Appendix~\ref{asymptotic}. We have thus derived an exact formula for the
Haar-averaged type-II error in discriminating a Haar-random measure-and-prepare
channel family $\{Q^{(N)}_U\}_{U \in U(d)}$ from the identity $\mathbb{I}$. The
result reduces to a unitary integral that can be evaluated using Weingarten
calculus. This first result uses the coherence of a maximally entangled input to
detect the loss of coherence caused by the random measurement channel. We now
turn to the complementary situation in which the observed data are already
classical measurement outcomes. In that regime the identity channel is no longer
the point of comparison; instead, the question is whether the outcome statistics
retain enough information to distinguish a shared Haar-random measurement from
independent blockwise Haar-random measurements.

\section{Discrimination via Haar-averaged measurement statistics}

In this section we study the classical statistical counterpart of the channel
discrimination problem above. The same operation, a Haar-random unitary followed
by a computational-basis measurement, now produces a classical outcome record.
We compare two random measurement models acting on a composite system of
$N=n_1+n_2$ subsystems of local dimension $d$. The models differ in the
structure of the underlying random channels: a single collective channel versus
two independent block-collective channels. Our goal is to characterize and
compare the Haar-averaged distributions of measurement outcomes induced by
these models, and to quantify their aggregate distinguishability using the total
variation distance. We also record how this statistic changes when
block-resolved histograms are available.

\subsection{Measurement channels}
We fix the computational basis $\{\ket{j}\}_{j=0}^{d-1}$ of $\mathbb{C}^d$ and
denote by
\begin{equation}
    \ket{\mathbf{0}} := \ket{0}^{\otimes N}
\end{equation}
the product input state on $\mathcal{H}^{\otimes N}$. Measurement outcomes are
labeled by strings
\begin{equation}
    \mathbf{k}=(k_0,\dots,k_{N-1})\in\{0,\dots,d-1\}^N,
\end{equation}
with the corresponding basis vectors
\begin{equation}
    \ket{\mathbf{k}}:=\ket{k_0}\otimes\cdots\otimes\ket{k_{N-1}}.
\end{equation}
We consider measurement channels obtained by applying a random unitary, followed
by a projective measurement in the computational basis. The classical
post-processing is described by the dephasing map $\Delta$
\begin{equation}
    \Delta(\rho):=\sum_{\mathbf{k}} \bra{\mathbf{k}}\rho\ket{\mathbf{k}}\ket{\mathbf{k}}\!\bra{\mathbf{k}},
\end{equation}
which discards all off-diagonal terms and produces a classical outcome distribution.
\begin{df}[Collective unitary measurement channel]
Let $N=n_1+n_2$. For fixed unitary $U \in U(d)$, define
\begin{equation}
    Q_U^{(N)}(\rho) := \Delta\left( U^{\otimes N}\,\rho\,U^{\dagger\otimes N} \right),
\end{equation}
and write
\begin{equation}
Q_U^{(N)}(\rho) = \sum_{\mathbf{k}} \mathrm{tr}\!\left(E_{U,\mathbf{k}}^{(N)}\,\rho\right) \ket{\mathbf{k}}\!\bra{\mathbf{k}}, \quad \text{where} \quad E_{U,\mathbf{k}}^{(N)}:= U^{\dagger\otimes N} \ket{\mathbf{k}}\!\bra{\mathbf{k}} U^{\otimes N}.
\end{equation}
\end{df}
Averaging over the Haar measure gives
\begin{equation}
\overline{Q}^{(N)}(\rho)=\sum_{\mathbf{k}}\mathrm{tr}\left( \overline{E}_{\mathbf{k}}^{(N)}\rho\right)\ket{\mathbf{k}}\!\bra{\mathbf{k}}, \qquad \overline{E}_{\mathbf{k}}^{(N)}:=\int E_{U,\mathbf{k}}^{(N)} dU. 
\end{equation}
\begin{df}[Two-block-collective unitary measurement channel]
For fixed $U,V \in U(d)$, define
\begin{equation}
Q_{U,V}^{(n_1,n_2)}(\rho):= \Delta\left[ \left(U^{\otimes n_1}\otimes V^{\otimes n_2}\right)\rho \left(U^{\dagger\otimes n_1}\otimes V^{\dagger\otimes n_2}\right) \right]
\end{equation}
where $U$ and $V$ are independent Haar-distributed unitaries in $U(d)$. Equivalently,
\begin{equation}
 Q_{U,V}^{(n_1,n_2)}(\rho) = \sum_{\mathbf{l},\mathbf{m}}\mathrm{tr}\!\left(E_{U, \mathbf{l}}^{(n_1)}\otimes E_{V,\mathbf{m}}^{(n_2)}\,\rho\right)\ket{\mathbf{l}\mathbf{m}}\!\bra{\mathbf{l}\mathbf{m}}.
\end{equation}
\end{df}
We write
\begin{equation}
    \overline{Q}^{(N)}:=\mathbb{E}_{U} \left[Q_U^{(N)}\right], \qquad \overline{Q}^{(n_1,n_2)}:=\mathbb{E}_{U,V} \left[Q_{U,V}^{(n_1,n_2)}\right]
\end{equation}
for the corresponding Haar-averaged channels. In the remainder of this section
we analyze the histogram distributions induced by $Q_U^{(N)}$ and
$Q_{U,V}^{(n_1,n_2)}$ on the input state $\ket{\mathbf{0}}$, and then average
these laws over the Haar measure. Exploiting permutation invariance, we first
derive an exact formula for the total variation distance between the resulting
Haar-averaged aggregate histogram laws, where the observer retains only the
total counts. We then compare this block-unresolved statistic with the
block-resolved statistic given by the ordered pair of histograms.

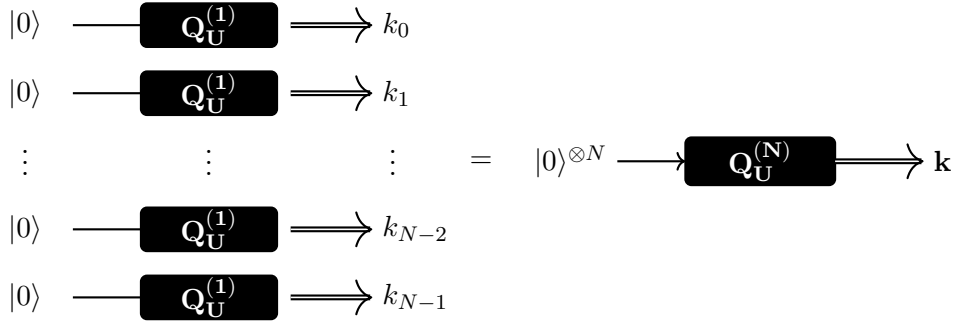
\begin{figure}[H]
\centering
\begin{tikzpicture}[scale=0.9]
\foreach \y/\lab in {
3/{|0\rangle},
2/{|0\rangle},
0/{|0\rangle},
-1/{|0\rangle}
}
{
    \node[left] at (-0.3,\y) {$\lab$};
}
\node at (-0.7,1.1) {$\vdots$};
\foreach \y in {3,2,0,-1}
{
    \draw[thick,->] (0,\y) -- (1.2,\y);
    \draw[thick, fill=black, rounded corners=0.08cm]
    (1.0,\y-0.32) rectangle +(2.0,0.64); 
    \node[color=white] at (2,\y)
    {$\mathbf{Q_U^{(1)}}$};
    \draw[double, thick, ->] (3.2,\y) -- (4.4,\y);
}
\node at (2.0,1.1) {$\vdots$};
\foreach \y/\lab in {
3/{k_0},
2/{k_1},
0/{k_{N-2}},
-1/{k_{N-1}}
}
{
    \node[right] at (4.4,\y) {$\lab$};
}
\node at (4.7,1.1) {$\vdots$};
\node at (6,1) {$=$};
\node[left] at (8,1)
{$|0\rangle^{\otimes N}$};
\draw[thick,->]
(8.0,1) -- (9,1);
\draw[thick, fill=black, rounded corners=0.08cm]
(9,0.65) rectangle +(2.2,0.7);
\node[color=white] at (10.1,1)
{$\mathbf{Q_U^{(N)}}$};
\draw[double, thick, ->]
(11.2,1) -- (12.5,1);
\node[right] at (12.5,1)
{$\mathbf{k}$};
\end{tikzpicture}
\caption{Equivalent representations of the collective model. The multipartite
implementation using $N$ copies of $Q_U^{(1)}$ is equivalent to the collective
channel $Q_U^{(N)}$ acting on the product input state $|0\rangle^{\otimes N}$,
producing the classical outcome vector $\mathbf{k}=(k_0,\dots,k_{N-1})$.}
\end{figure}

\begin{figure}[H] 
\centering 
\begin{tikzpicture}[scale=0.9]
\foreach \y/\lab in {
3/{|0\rangle},
2/{|0\rangle},
0/{|0\rangle},
-1/{|0\rangle}
}
{
    \node[left] at (-0.3,\y) {$\lab$};
}
\foreach \y in {3,2}
{
    \draw[thick,->] (0,\y) -- (1.2,\y);
    \draw[thick, fill=black, rounded corners=0.08cm]
    (1.2,\y-0.32) rectangle +(2.0,0.64);
    \node[color=white] at (2.2,\y)
    {$\mathbf{Q_U^{(1)}}$};
    \draw[double, thick, ->] (3.2,\y) -- (4.4,\y);
}  
\node[left] at (-0.5,1.1) {$\vdots$};
\node at (2.2,1.1) {$\vdots$};
\node[right] at (4.8,1.1) {$\vdots$};
\foreach \y in {0,-1}
{
    \draw[thick,->] (0,\y) -- (1.2,\y);
    \draw[thick, fill=gray!50, rounded corners=0.08cm]
    (1.2,\y-0.32) rectangle +(2.0,0.64);
    \node at (2.2,\y)
    {$\mathbf{Q_V^{(1)}}$};
    \draw[double, thick, ->] (3.2,\y) -- (4.4,\y);
}
\foreach \y/\lab in {
3/{k_0},
2/{k_1},
0/{k_{N-2}},
-1/{k_{N-1}}
}
{
    \node[right] at (4.7,\y) {$\lab$};
}
\node at (6.8,0.8) {$=$};
\node[left] at (9.0,1.4)
{$|0\rangle^{\otimes n_1}$};
\draw[thick,->]
(9.1,1.4) -- (10.5,1.4);
\draw[thick, fill=black, rounded corners=0.08cm]
(10.5,1.05) rectangle +(2.5,0.7);
\node[color=white] at (11.75,1.4)
{$\mathbf{Q_U^{(n_1)}}$};
\draw[double, thick, ->]
(13,1.4) -- (14.2,1.4);
\node[right] at (14.2,1.4)
{$\mathbf{k}^{(1)}$};
\node[left] at (9.0,0)
{$|0\rangle^{\otimes n_2}$};
\draw[thick,->]
(9.1,0) -- (10.5,0);
\draw[thick, fill=gray!50, rounded corners=0.08cm]
(10.5,-0.35) rectangle +(2.5,0.7);
\node at (11.75,0)
{$\mathbf{Q_V^{(n_2)}}$};
\draw[double, thick, ->]
(13,0) -- (14.2,0);
\node[right] at (14.2,0)
{$\mathbf{k}^{(2)}$};
\end{tikzpicture}
\caption{Equivalent representations of the block model. The multipartite
implementation using $n_1$ copies of $Q_U^{(1)}$ and $n_2$ copies of $Q_V^{(1)}$
is equivalent to two collective channels $Q_U^{(n_1)}$ and $Q_V^{(n_2)}$ acting
independently on the corresponding subsystems, producing the classical outcome
vectors \(\mathbf{k}^{(1)}\) and \(\mathbf{k}^{(2)}\). In the aggregate experiment, the observer
retains only the total histogram obtained after forgetting the block assignment
of the outcomes.}
\end{figure}
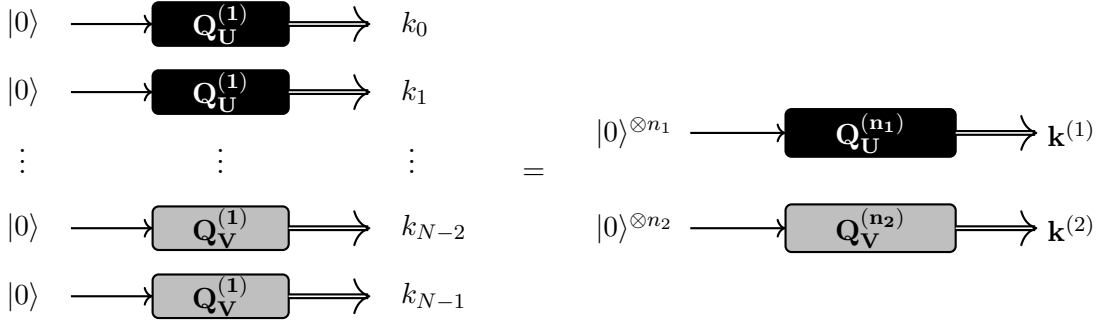

\subsection{Histogram reduction for a single random unitary}
In this section, we show that the output statistics of the random measurement
channel can be reduced to a distribution over histograms, which significantly
simplifies the analysis. 

\begin{lem}[Permutation invariance]
\label{permutacje}
Let $\mathbf{k},\mathbf{k}'\in\{0,\dots,d-1\}^N$ be two outcome strings such
that $\mathbf{k}'$ is obtained from $\mathbf{k}$ by a permutation of subsystems.
Then
\begin{equation}
\mathrm{tr}\!\left(\overline{E}_{\mathbf{k}}^{(N)}\ket{\mathbf{0}}\!\bra{\mathbf{0}}\right) = \mathrm{tr}\!\left(\overline{E}_{\mathbf{k}'}^{(N)}\ket{\mathbf{0}}\!\bra{\mathbf{0}}\right).
\end{equation}
Consequently, the probability of observing an outcome $\mathbf{k}$ depends only
on the multiplicities of the symbols appearing in $\mathbf{k}$.
\end{lem}
\begin{proof}
Let $\sigma\in S_N$ be a permutation such that $\mathbf{k}'=\sigma(\mathbf{k})$,
and denote by $P_\sigma$ the corresponding permutation operator on
$\mathcal{H}^{\otimes N}$. Using the Haar invariance and the fact that
$P_\sigma\ket{\mathbf{0}}=\ket{\mathbf{0}}$, we obtain
\begin{equation}
\overline{E}_{\mathbf{k}'}^{(N)}=P_\sigma\,\overline{E}_{\mathbf{k}}^{(N)}\,P_\sigma^\dagger,
\end{equation}
which implies
\begin{equation}
\mathrm{tr}\!\left(\overline{E}_{\mathbf{k}}^{(N)}\ket{\mathbf{0}}\!\bra{\mathbf{0}}\right) = \mathrm{tr}\!\left(\overline{E}_{\mathbf{k}'}^{(N)}\ket{\mathbf{0}}\!\bra{\mathbf{0}}\right).
\end{equation}
The detailed permutation-based argument is given in Appendix~\ref{permutacjeobl}.
\end{proof}
Motivated by Lemma~\ref{permutacje}, we introduce the histogram associated with
an outcome string $\mathbf{k}$. We denote by $\Lambda$ the random histogram of
measurement outcomes obtained from the channel $Q_U^{(N)}$ acting on the input
state $|\mathbf{0}\rangle$. Thus, $\Lambda$ is a random variable taking values
in the set of integer vectors $\lambda=(\lambda_0,\dots,\lambda_{d-1}) \in
\mathbb{N}^d$ such that $\sum_{j=0}^{d-1} \lambda_j=N.$ A particular realization
of this random variable will be denoted by $\lambda$.
\begin{df}[Outcome histogram]
Given an outcome string $\mathbf{k}=(k_0,\dots,k_{N-1})$, the associated
histogram is the vector
\begin{equation}
    \begin{split}
        & \lambda=(\lambda_0,\dots,\lambda_{d-1}), \\
        &\lambda_j:=\#\{r: k_r=j\},
    \end{split}
\end{equation}
satisfying $\sum_{j=0}^{d-1}\lambda_j=N$.
\end{df}

\begin{lem}[Conditional multinomial distribution]
\label{multi}
Conditioned on a fixed unitary $U$, the histogram $\Lambda$ of outcomes produced
by $Q_U^{(N)}$ acting on $\ket{\mathbf{0}}$ follows a multinomial distribution
with parameters $(N,W)$, where
\begin{equation*}
W=(W_0,W_1,\dots,W_{d-1}), \quad \text{and} \quad W_j := |\langle j|U|0\rangle|^2 \quad\text{are the probabilities of observing outcome $j$}.
\end{equation*}
That is,
\begin{equation}
\mathbb{P}(\Lambda=\lambda\,|\,U) = \frac{N!}{\prod_{j=0}^{d-1}\lambda_j!} \prod_{j=0}^{d-1} W_j^{\lambda_j}.
\end{equation}
\end{lem}
Since the measurement outcomes are independent and identically distributed
according to the probability vector $W$, the claim follows directly from
standard multinomial statistics. A detailed derivation is provided in
Appendix~\ref{multinomial}.
\begin{propo}[Haar-averaged histogram distribution]
\label{rozkladPN}
Let $\mathbb{P}(\Lambda=\lambda)$ denote the Haar-averaged distribution of the
histogram $\Lambda$ obtained from $Q_U^{(N)}$ acting on $\ket{\mathbf{0}}$. Then
$\mathbb{P}(\Lambda =\lambda)$ is uniform over all compositions of $N$ into $d$
parts:
\begin{equation}
\mathbb{P}(\Lambda =\lambda)=\frac{(d-1)!N!}{(N+d-1)!}.
\end{equation}
\end{propo}
\begin{proof}
We compute the Haar average:
\begin{equation}
    \begin{split}
        & \mathbb{P}(\Lambda =\lambda) = m(\lambda) \mathbb{E}_U \left[ \prod_{j=0}^{d-1} W_j^{\lambda_j} \right] = m(\lambda) \int \prod_{j=0}^{d-1} W_j^{\lambda_j} dU.\\
    \end{split}
\end{equation}
where $m(\lambda)$ is the number of distinct sequences $\mathbf{k}$ with counts
specified by $\lambda$, i.e.
\begin{equation}
m(\lambda):=\binom{N}{\lambda_0,\dots,\lambda_{d-1}}=\frac{N!}{\prod_j \lambda_j!}.
\end{equation}
For a fixed vector $\lambda = (\lambda_0,\dots,\lambda_{d-1})$, denote
\begin{equation}
I(\lambda) := \int \prod_{j=0}^{d-1} W_j^{\lambda_j} dU .
\end{equation}
For a Haar-distributed unitary $U$, the probability vector
$W=(W_0,\dots,W_{d-1})$ is uniformly distributed over the probability simplex.
Thus, the vector $W$ has the Dirichlet distribution $\text{Dir}(1,\dots,1)$ on
the simplex $\Delta_{d-1}$, since under the Haar measure all probability
configurations are equally likely. Here
\begin{equation}
\Delta_{d-1} = \{ p \in \mathbb{R}^d: p_i \ge 0, \sum_{i=0}^{d-1} p_i =1\}.
\end{equation}
Therefore
\begin{equation}
\mathbb{P}(\Lambda =\lambda) = m(\lambda)I(\lambda)=\frac{(d-1)!N!}{(N+d-1)!}.
\end{equation}
The explicit value of $I(\lambda)$ is calculated in Appendix~\ref{dirichlet}.
\end{proof}
Since the distribution is uniquely determined, we denote
\begin{equation}
P^{(N)}(\lambda):=\mathbb{P}(\Lambda =\lambda).
\end{equation}

\begin{lem}[Blockwise Haar–averaged histogram distribution] \label{wspolny}
Consider the random measurement channel $Q_{U,V}^{(n_1,n_2)}$ acting on the product input state
\begin{equation}
\ket{\mathbf{0}}=\ket{0}^{\otimes n_1}\otimes \ket{0}^{\otimes n_2}.
\end{equation}
Let $\lambda^{(1)}$ and $\lambda^{(2)}$ denote the outcome histograms associated
with the first and second block, respectively. Then the following holds:
\begin{enumerate}
    \item Conditioned on fixed unitaries $U$ and $V$, the outcomes within each
    block are independent and identically distributed, and the corresponding
    histograms follow multinomial distributions.

    \item After averaging over the Haar measure, the joint distribution of
    $(\lambda^{(1)},\lambda^{(2)})$ factorizes because the two Haar-distributed unitaries are independent:
    \begin{equation}
    \mathbb{P}(\Lambda^{(1)}=\lambda^{(1)},\Lambda^{(2)}=\lambda^{(2)})
    = P^{(n_1)}(\lambda^{(1)}) \cdot P^{(n_2)}(\lambda^{(2)})=\frac{n_1!n_2!(d-1)!(d-1)!}{(n_1+d-1)!(n_2+d-1)!},
    \end{equation}
    for all histograms satisfying $ \sum_j \lambda^{(1)}_j=n_1 $ and $\sum_j \lambda^{(2)}_j=n_2.$
\end{enumerate}
In particular, all admissible pairs of histograms occur with equal probability.
\end{lem}
\begin{proof}
The result follows from the permutation invariance of the Haar measure and the
independence of the two unitary blocks. A detailed derivation, including the
multinomial structure and the Haar integration leading to the Dirichlet
distribution, is given in Appendix~\ref{Pn1Pn2}.
\end{proof}
\begin{propo}[Histogram decompositions]
\label{wyliczenieL}
Fix a histogram $\lambda=(\lambda_0,\dots,\lambda_{d-1})$ with
$\sum_{j=0}^{d-1}\lambda_j=N$. Let $L_{\lambda}(n_1,n_2)$ denote the number of
vectors
\begin{equation}
\lambda^{(1)}=(k_0,\dots,k_{d-1})\in\mathbb{N}^d
\end{equation}
satisfying
\begin{equation}
\sum_{j=0}^{d-1} k_j = n_1, \quad \text{and} \quad 0\le k_j \le \lambda_j\ \text{for all } j.
\end{equation}
Equivalently, $L_{\lambda}(n_1,n_2)$ counts the number of ways in which the
histogram $\lambda$ can be decomposed as $\lambda=\lambda^{(1)}+\lambda^{(2)}$
with $\sum_j \lambda^{(1)}_j=n_1$ and $\sum_j \lambda^{(2)}_j=n_2$. Then
$L_{\lambda}(n_1,n_2)$ admits the inclusion–exclusion representation
\begin{equation}
\label{wzorL}
L_{\lambda}(n_1,n_2)=\sum_{S\subseteq\{0,\dots,d-1\}}(-1)^{|S|}\binom{n_1-\sum_{j\in S}(\lambda_j+1)+d-1}{d-1},
\end{equation}
where $\binom{a}{b}=0$ for $a<b$.
\end{propo}
The formula follows from a standard inclusion–exclusion argument for integer
compositions with upper bounds. Details are given in Appendix~\ref{liczbapar}.

\begin{propo}[Effective aggregate histogram distribution]
\label{lambdarozklad}
Let $P^{(n_1,n_2)}(\lambda)$ denote the Haar-averaged distribution of the total
outcome histogram $\Lambda$ of size $N=n_1+n_2$ obtained from
$Q_{U,V}^{(n_1,n_2)}$ acting on the input state
$\ket{\mathbf{0}}=\ket{0}^{\otimes n_1} \otimes \ket{0}^{\otimes n_2}$. This is
the block-unresolved distribution obtained after the coarse-graining
\begin{equation}
(\lambda^{(1)},\lambda^{(2)})\longmapsto \lambda=\lambda^{(1)}+\lambda^{(2)}.
\end{equation}
Then
\begin{equation}
    \begin{split}
        & P^{(n_1,n_2)}(\lambda):=\frac{n_1!n_2!(d-1)!(d-1)!}{(n_1+d-1)!(n_2+d-1)!}L_{\lambda}(n_1,n_2), \\
    \end{split}
\end{equation}
where $L_{\lambda}(n_1,n_2)$ is defined in Proposition~\ref{wyliczenieL}.
\end{propo}
\begin{proof}
We can obtain the effective distribution $P^{(n_1,n_2)}(\lambda)$ for $\lambda =
\lambda^{(1)}+\lambda^{(2)}$ by summing over all admissible pairs
$(\lambda^{(1)},\lambda^{(2)})$ satisfying this constraint:
\begin{equation}
    \begin{split}
        &P^{(n_1,n_2)}(\lambda) = \sum_{\lambda^{(1)},\lambda^{(2)}} 
        \mathbf{1}_{\{\lambda=\lambda^{(1)}+\lambda^{(2)}\}} \mathbb{P}(\Lambda^{(1)}=\lambda^{(1)}, \Lambda^{(2)}=\lambda^{(2)}).
    \end{split}
\end{equation}
Since each pair contributes the same constant value, we obtain from
Lemma~\ref{wspolny}
\begin{equation}
    \begin{split}
        & P^{(n_1,n_2)}(\lambda) = \frac{n_1!n_2!(d-1)!(d-1)!}{(n_1+d-1)!(n_2+d-1)!}L_{\lambda}(n_1,n_2). \\
    \end{split}
\end{equation}
This quantity counts the number of ways in which the total histogram $\lambda$
can be decomposed into two component histograms $\lambda^{(1)}$ and
$\lambda^{(2)}$ of sizes $n_1$ and $n_2$, respectively.
\end{proof}

\subsection{Total variation distance}
\begin{cor}[Aggregate total variation distance]
We quantify the discrepancy between the single–block distribution
$P^{(N)}(\lambda)$ and the block–separated aggregate distribution
$P^{(n_1,n_2)}(\lambda)$ by the total variation distance
\begin{equation}
\mathrm{TVD}\left(P^{(N)},P^{(n_1,n_2)}\right)=\frac12\sum_{\lambda}\big|P^{(N)}(\lambda)-P^{(n_1,n_2)}(\lambda)\big|.
\end{equation}
Let
\begin{equation}
M=\binom{N+d-1}{d-1},\qquad M_1=\binom{n_1+d-1}{d-1},\qquad M_2=\binom{n_2+d-1}{d-1}.
\end{equation}
and define
\begin{equation}
\overline{M}=\frac{M_1M_2}{M}.
\end{equation}
Using Proposition~\ref{lambdarozklad} and the fact that
\begin{equation}
P^{(N)}(\lambda)=\frac{1}{M}, \qquad P^{(n_1,n_2)}(\lambda)=\frac{L_{\lambda}(n_1,n_2)}{M_1M_2}, \qquad \frac{P^{(n_1,n_2)}(\lambda)}{P^{(N)}(\lambda)}=\frac{L_{\lambda}(n_1,n_2)}{\overline{M}}
\end{equation}
we obtain
\begin{equation}
\mathrm{TVD}\left(P^{(N)},P^{(n_1,n_2)}\right)=\frac12 \mathbb{E}_{\lambda \sim P^{(N)}}\left[\left| 1- \frac{L_{\lambda}(n_1,n_2)}{\overline{M}}\right|\right].
\end{equation}
\end{cor}
Although both distributions admit explicit formulas, the evaluation of the
aggregate total variation distance requires averaging the combinatorial quantity
$L_{\lambda}(n_1,n_2)$ over all histograms. This motivates the asymptotic
analysis carried out in the next section.

\subsection{Block-resolved histograms and operational interpretation}
\label{blockresolved}

The preceding aggregate total variation distance refers to the statistic
\begin{equation}
\lambda=\lambda^{(1)}+\lambda^{(2)}.
\end{equation}
This is the natural statistic when the measurement record preserves outcome
labels but not the block of origin of each occurrence. If the observer instead
has access to the ordered pair of block histograms
$(\lambda^{(1)},\lambda^{(2)})$, then the relevant experiment is different and
generally more informative.

Let
\begin{equation}
\mathcal C_{m,d}:=\left\{a\in\mathbb N^d:\sum_{j=0}^{d-1}a_j=m\right\},\qquad
|\mathcal C_{m,d}|=\binom{m+d-1}{d-1}.
\end{equation}
For $a\in\mathcal C_{n_1,d}$ and $b\in\mathcal C_{n_2,d}$, write $\lambda=a+b$.
In the shared-unitary model, corresponding to the collective channel split into
two observed blocks, both histograms are sampled from the same Haar-random
probability vector
\begin{equation}
W=(W_0,\dots,W_{d-1})\sim\mathrm{Dirichlet}(1,\dots,1).
\end{equation}
Conditioned on $W$, the two block histograms are independent multinomials with
parameters $(n_1,W)$ and $(n_2,W)$:
\begin{equation}
P(a,b\mid W)=\left(\frac{n_1!}{\prod_j a_j!}\prod_j W_j^{a_j}\right)\left(\frac{n_2!}{\prod_j b_j!}\prod_j W_j^{b_j}\right).
\end{equation}
Hence
\begin{equation}
P(a,b\mid W)=\frac{n_1!n_2!}{\prod_j a_j!b_j!}\prod_{j=0}^{d-1}W_j^{a_j+b_j}.
\end{equation}
Averaging over the Haar measure and using the Dirichlet integral from
Appendix~\ref{dirichlet}, we obtain
\begin{equation}
\int_{\Delta_{d-1}}
\prod_j W_j^{a_j+b_j}\,dU
=
(d-1)!
\frac{\prod_j (a_j+b_j)!}{(N+d-1)!}.
\end{equation}
Therefore
\begin{equation}
\label{pairshared}
P_{\mathrm{sh}}(a,b)
=\frac{n_1!n_2!(d-1)!}{(N+d-1)!}
\prod_{j=0}^{d-1}\frac{(a_j+b_j)!}{a_j!b_j!}.
\end{equation}
Equivalently,
\begin{equation}
P_{\mathrm{sh}}(a,b)=\frac{1}{M}\,\frac{\prod_j\binom{a_j+b_j}{a_j}}{\binom{N}{n_1}}.
\end{equation}
In the independent-block model, the two Haar averages factorize. By the uniform
histogram law derived earlier, each block histogram is uniformly distributed on
its own composition set $P(a)=\frac1{M_{1}}$, $P(b)=\frac1{M_{2}}$. Since the two blocks are independent,
\begin{equation}
\label{pairind}
P_{\mathrm{ind}}(a,b)=\frac{1}{M_{1}M_{2}}.
\end{equation}
Thus the block-resolved total variation distance is
\begin{equation}
\label{pairtvd}
\mathrm{TVD}\left(P_{\mathrm{sh}},P_{\mathrm{ind}}\right)=\frac12 \sum_{a,b}\left|P_{\mathrm{sh}}(a,b)-P_{\mathrm{ind}}(a,b)\right|=\frac12\sum_{a\in\mathcal C_{n_1,d}}\sum_{b\in\mathcal C_{n_2,d}}\left|\frac{1}{M}\frac{\prod_j\binom{a_j+b_j}{a_j}}{\binom{N}{n_1}}-\frac{1}{M_{1}M_{2}}\right|.
\end{equation}
This distance is also the operational distance for full block-labeled outcome
strings after Haar averaging: both models are exchangeable within each block, so
conditional on the pair $(a,b)$ all strings with those histograms have the same
probability under both hypotheses, and the likelihood ratio depends only on
$(a,b)$. Since under the independent-block model the pair $(a,b)$ is uniformly
distributed on $\mathcal C_{n_1,d}\times \mathcal C_{n_2,d}$, we may rewrite the
total variation distance as an expectation with respect to the uniform law.
\begin{equation}
\label{pairtvdlr}
\mathrm{TVD}\left(P_{\mathrm{sh}},P_{\mathrm{ind}}\right)
=\frac12\,\mathbb E_{a,b\,\mathrm{uniform}}\left[\left|\frac{M_{1}M_{2}}{M\binom{N}{n_1}} \prod_{j=0}^{d-1}\binom{a_j+b_j}{a_j}-1\right|\right].
\end{equation}
The aggregate formulas above are recovered by applying the coarse-graining map
$(a,b)\mapsto a+b$. Indeed, by Vandermonde's identity,
\begin{equation}
\sum_{a+b=\lambda}\prod_{j=0}^{d-1}\binom{\lambda_j}{a_j}=\binom{N}{n_1},
\end{equation}
and therefore
\begin{equation}
\sum_{a+b=\lambda}P_{\mathrm{sh}}(a,b)=\frac{1}{M}=P^{(N)}(\lambda),
\qquad
\sum_{a+b=\lambda}P_{\mathrm{ind}}(a,b)=\frac{L_\lambda(n_1,n_2)}{M_{1}M_{2}}=P^{(n_1,n_2)}(\lambda).
\end{equation}
The map $(a,b)\mapsto a+b=\lambda$ forgets the information about which block
produced a given outcome count. In other words, it is a coarse-graining from the
pair histogram to the aggregate histogram. Since total variation distance cannot
increase under deterministic post-processing, we obtain
\begin{equation}
\label{dataprocessing}
\mathrm{TVD}\left(P^{(N)},P^{(n_1,n_2)}\right)\le \mathrm{TVD}\left(P_{\mathrm{sh}},P_{\mathrm{ind}}\right).
\end{equation}
Thus the total-histogram TVD studied in the sequel is a block-unresolved
distinguishability, or equivalently a lower bound on the distinguishability
available from block-resolved data.

The extra information in $(a,b)$ can also be seen conditionally on the total
histogram. Under the shared-unitary model,
\begin{equation}
P_{\mathrm{sh}}(a\mid \lambda)
=\frac{\prod_j\binom{\lambda_j}{a_j}}{\binom{N}{n_1}},
\end{equation}
which is a multivariate hypergeometric law: the first block looks like a random
subset of $n_1$ samples from the $N$ aggregate samples. Under the
independent-block model,
\begin{equation}
P_{\mathrm{ind}}(a\mid \lambda)=\frac{1}{L_\lambda(n_1,n_2)}.
\end{equation}
Hence the two conditional laws coincide for collision-free histograms but differ
once some $\lambda_j\ge2$. For instance, if $\lambda_j=2$, the shared-unitary
model assigns relative weights $1,2,1$ to the local decompositions
$(a_j,b_j)=(2,0),(1,1),(0,2)$, while the independent-block model weights
admissible decompositions uniformly. The pair statistic therefore detects
whether a collision occurred inside a single block or across the two blocks,
whereas the aggregate statistic only records that a collision occurred.

\section{\texorpdfstring{Asymptotics of TVD between $P^{(N)}$ and
$P^{(n_1,n_2)}$}{Asymptotics of TVD between P(N) and P(n1,n2)}}

\subsection{Overview of asymptotic regimes}

In this section we investigate the asymptotic behaviour of the aggregate total
variation distance between the Haar–averaged total-histogram laws
$P^{(N)}(\lambda)$ and $P^{(n_1,n_2)}(\lambda)$ in the joint limit
$N,d\to\infty$. Thus the statistic in this section is the block-unresolved
histogram $\lambda=\lambda^{(1)}+\lambda^{(2)}$. Depending on the relative
growth of the number of copies $N$ and the Hilbert space dimension $d$, we
distinguish several asymptotic regimes.

\begin{itemize}
    \item for fixed $N$ and $d\to\infty$, we show that the aggregate total
    variation distance vanishes asymptotically with leading term
    $\frac{n_1n_2}{d}$.
    \item for fixed $d$ and $N\to\infty$, we show that the aggregate total
    variation distance converges to a nontrivial simplex-limit depending on $d$
    and $\alpha=\frac{n_1}{N}$.
    \item In the sparse joint scaling regime $N=o(\sqrt{d})$, we recover the
    same leading collision-driven asymptotic $\frac{n_1n_2}{d}$,
    \item in the critical scaling regime $N/\sqrt d\to c$, we obtain an exact
    limiting formula for the aggregate $\mathrm{TVD}$ in terms of a
    Poisson-distributed collision count, leading to a closed probabilistic
    expression for the distinguishability. 
\end{itemize}
Each regime is analyzed separately in the subsections below. Depending on the
relative scaling of $N$ and $d$, the distinguishability exhibits qualitatively
different asymptotic behaviours. The regimes analyzed in this section are
summarized in the table below.

\begin{center}
\begin{tabular}{|c|c|c|c|} 
\hline
\textbf{Statistic} & \textbf{Regime} & \textbf{Limiting behaviour} & \textbf{Main result}  \\
\hline
Aggregate $\mathrm{TVD}$ & fixed $N$, $d \to \infty $& $\frac{n_1n_2}{d}+O(d^{-2})$ & Proposition~\ref{fixedN} \\
\hline
Aggregate $\mathrm{TVD}$ & fixed $d$, $N\to \infty $& simplex integral & Proposition~\ref{fixedd}\\
\hline
Aggregate $\mathrm{TVD}$ & $N=o(\sqrt{d})$& $\frac{n_1n_2}{d}+O\left(\frac{N^4}{d^2}\right)$ & Proposition~\ref{JSsparse} \\
\hline
Aggregate $\mathrm{TVD}$ & $N/\sqrt d\to c$& \makecell{$\frac{1}{2}\mathbb{E} \left| 1 - e^{\alpha(1-\alpha) c^2}(1-\alpha(1-\alpha))^K \right|$\\ $K\sim\mathrm{Poisson(c^2)}$} & Theorem~\ref{twCritScal} \\
\hline
Block-resolved $\mathrm{TVD}$ & fixed $N$, $d \to \infty $& $\frac{n_1n_2}{d}+O(d^{-2})$& Proposition~\ref{SparseResolved} \\
\hline
Block-resolved $\mathrm{TVD}$ & $N=o(\sqrt{d})$& $\frac{n_1n_2}{d}+O\left(\frac{N^4}{d^2}\right)$ & Proposition~\ref{SparseResolved} \\
\hline
Block-resolved $\mathrm{TVD}$ & $N/\sqrt d\to c$& \makecell{$\frac12\,\mathbb E\left[\left|e^{-\alpha(1-\alpha)c^2}2^C-1\right|\right]$\\ $C\sim \mathrm{Poisson}(\alpha(1-\alpha)c^2)$} & Theorem~\ref{CriticalResolved} \\
\hline
Block-resolved $\mathrm{TVD}$ & fixed $d$, $N\to \infty $& tends to $1$& Corollary~\ref{ResolvedFixedd}\\
\hline
\end{tabular}
\end{center}
Throughout this section we use the representation of the aggregate total
variation distance in terms of the histogram functional $L_{\lambda}(n_1,n_2)$
derived in the previous section. Recall that
\begin{equation}
\mathrm{TVD}\left(P^{(N)},P^{(n_1,n_2)}\right)=\frac12\,\mathbb E_{\lambda\sim P^{(N)}}\!\left[\left|1-\frac{L_{\lambda}(n_1,n_2)}{\overline{M}}\right|\right],\qquad\overline{M}=\frac{M_1M_2}{M},
\end{equation}
where the expectation is taken with respect to the uniform distribution over stars–and–bars histograms and 
\begin{equation}
M=\binom{N+d-1}{d-1},\qquad M_1=\binom{n_1+d-1}{d-1},\qquad M_2=\binom{n_2+d-1}{d-1}.
\end{equation}

\subsection{\texorpdfstring{Fixed $N$, large $d$ for aggregate model}{Fixed N, large d for aggregate model}} \label{collision}

We analyze the behavior of the aggregate total variation distance between the
two Haar–averaged sampling models:
\begin{equation*}
\mathrm{TVD}\left(P^{(N)},P^{(n_1,n_2)}\right),
\end{equation*}
in the limit of large local Hilbert-space dimension $d$, while $N$ is fixed. Due
to symmetry under Haar averaging, the probability of an $N$-shot aggregate
measurement record depends only on its histogram
$\lambda=(\lambda_0,\dots,\lambda_{d-1})$ of occupation numbers. We evaluate the
aggregate TVD by separating histograms according to their collision structure,
because collision events constitute the dominant source of distinguishability
between the two models. We split histograms by the number of repeated
occupations, distinguishing collision-free and single-collision cases:
\begin{equation}
    \begin{split}
        &S_0 = \{\lambda:\lambda_j\in\{0,1\},\ \sum_j\lambda_j=N\},\\
        &S_1 = \{\lambda:\text{exactly one} ,
        \lambda_j=2,\text{ all others in }\{0,1\}\}.
    \end{split}
\end{equation}
Histograms with at least two collisions contribute only $O(d^{-2})$ to the aggregate $\mathrm{TVD}$.
Let
\begin{equation}
S_{\ge2}=\{\lambda :\text{$\lambda$ contains at least two collision pairs}\}.
\end{equation}
For fixed $N$, every additional collision reduces by one the number of freely chosen occupied bins. Hence
\begin{equation}
|S_{\ge2}| = O(d^{N-2}).
\end{equation}
Moreover, uniformly in $\lambda$,
\begin{equation}
P^{(N)}(\lambda)=O(d^{-N}),\qquad
P^{(n_1,n_2)}(\lambda)=O(d^{-N}).
\end{equation}
Therefore
\begin{equation}
\sum_{\lambda\in S_{\ge2}}\left|P^{(N)}(\lambda)-P^{(n_1,n_2)}(\lambda)\right|=O(d^{N-2})O(d^{-N})=O(d^{-2}).
\end{equation}
\begin{propo}[Fixed $N$, large $d$]\label{fixedN} Fix integers $N\ge1$ and
    $n_1,n_2\ge0$ with $n_1+n_2=N$. As $d\to\infty$, we obtain
    \begin{equation}
\mathrm{TVD}\left(P^{(N)},P^{(n_1,n_2)}\right)= \mathrm{TVD}_{S_0} + \mathrm{TVD}_{S_1} + O(d^{-2}) = \frac{n_1n_2}{d}+O(d^{-2})
    \end{equation}
\end{propo}
\begin{proof}
We decompose the set of histograms according to the number of collisions.
Collision-free histograms $S_0$ and single-collision histograms $S_1$ give the
leading contributions, while all histograms outside $S_0\cup S_1$ contribute
only $O(d^{-2})$. The detailed combinatorial evaluation is given in
Appendix~\ref{proofpropo}.
\end{proof}
Hence, the aggregate distinguishability between collective and block-separated
Haar-random measurements decays inversely with $d$, and is proportional to the
number of cross-block outcome pairs $n_1n_2$, reflecting the suppression of
global correlations in the large-dimension limit.

\subsection{\texorpdfstring{Fixed $d$, large $N$: simplex limit ($n_1=\lfloor
\alpha N \rfloor$ and $\alpha \in (0,1)$ fixed) for aggregate model}{Fixed d, large N: simplex limit
(n1=floor(alpha N) and alpha in (0,1) fixed) for aggregate model}}

In this section we study the asymptotic behavior of the aggregate total
variation distance in the regime where the local dimension $d$ is fixed and the
number of measurement outcomes $N$ tends to infinity. We consider a splitting
$n_1=\lfloor \alpha N \rfloor$, $n_2=N-n_1$ with $\alpha\in(0,1)$ fixed. As $N
\to \infty$ with $d$ fixed, the aggregate total variation distance admits a
finite limit depending only on $d$ and $\alpha$. We denote this limit by
\begin{equation}
    \mathrm{TVD}_{\infty}(d,\alpha):= \lim_{N \to \infty}\mathrm{TVD}\left(P^{(N)},P^{(n_1,n_2)}\right).
\end{equation}
Below we show that this limit exists and can be expressed explicitly. We then
evaluate the general limiting formula for low-dimensional cases.

\subsubsection{\texorpdfstring{Asymptotic expansion of
$L_{\lambda}(n_1,n_2)$}{Asymptotic expansion of L(lambda)(n1,n2)}}

For fixed $\lambda$ and large $N$, $L_{\lambda}(n_1,n_2)$ counts the lattice points in a truncated simplex.
\begin{equation}
\mathcal{S}(\lambda)=\{x\in\mathbb R_{\ge0}^d: \sum_j x_j=n_1,\; x_j\le\lambda_j\}
\end{equation}
Let $\ell=\frac{\lambda}{N} \in\Delta_{d-1}$, $t=\frac{x}{n_1}\in\Delta_{d-1}$. Then $\mathcal{S}(\lambda)$ is $n_1^{d-1}$
times the fixed polytope
\begin{equation}
\left\{t\in\Delta_{d-1}: t_j\le u_j(\ell):=\min\left(1,\frac{\ell_j}{\alpha}\right)\right\}.
\end{equation}
It is convenient to express $L_{\lambda}(n_1,n_2)$ in terms of generating
functions. Writing $\lambda=(\lambda_0,\dots,\lambda_{d-1})$, we have
\begin{equation}
L_{\lambda}(n_1,n_2) = [x^{n_1}] \prod_{j=0}^{d-1} (1 + x + \cdots + x^{\lambda_j}).
\end{equation}
This representation allows us to recover the inclusion–exclusion formula and
provides an alternative route to the asymptotic analysis used below.
\begin{lem}[Truncated-simplex asymptotics]\label{Postacpsi}
Fix $d\ge2$ and $\alpha\in(0,1)$. Let $N\to\infty$,
$n_1=\lfloor\alpha N\rfloor$, $n_2=N-n_1$, and let
$\lambda=\lambda(N)$ be any histogram with $\sum_j\lambda_j=N$ and
$\ell=\lambda/N$. Then $L_{\lambda}(n_1,n_2)$ admits the asymptotic expansion
\begin{equation}
L_{\lambda}(n_1,n_2)=\frac{n_1^{d-1}}{(d-1)!}\,\varphi_\alpha(\ell)+O(n_1^{d-2}),
\end{equation}
uniformly in $\lambda$, where
\begin{equation}
\varphi_\alpha(\ell)=\sum_{S\subseteq[d]}(-1)^{|S|}\left(1-\sum_{j\in S}\min\left\{1,\frac{\ell_j}{\alpha}\right\}\right)_{+}^{d-1}
\end{equation}
is the normalized $(d-1)$-dimensional volume of the continuous polytope obtained
by scaling $x$ and $\lambda$. Its explicit inclusion–exclusion expression is
derived below.
\end{lem}
\begin{proof}
See Appendix~\ref{erhart}
\end{proof}

\begin{propo}[Fixed $d$, large $N$]\label{fixedd}
    Fix $d\ge2$ and $\alpha\in(0,1)$. Let $N\to\infty$,
    $n_1=\lfloor\alpha N\rfloor$, and $n_2=N-n_1$. Then
    \begin{equation}
        \mathrm{TVD}\left(P^{(N)},P^{(n_1,n_2)}\right)= \frac12 \mathbb E_{\ell\sim \mathrm{Unif}(\Delta_{d-1})} \left[\left|1-\frac{\varphi_\alpha(\ell)}{(1-\alpha)^{d-1}}\right|\right] +o(1).
    \end{equation}
    In particular, the above limit exists and is given by
    \begin{equation}
        \mathrm{TVD}_{\infty}(d,\alpha)=\frac12 \mathbb E_{\ell\sim \mathrm{Unif}(\Delta_{d-1})} \left[\left|1-\frac{\varphi_\alpha(\ell)}{(1-\alpha)^{d-1}}\right|\right].
    \end{equation}
\end{propo}
\begin{proof} 
To pass to the limit, we use the representation of the aggregate total variation
distance as an average over histograms
\begin{equation}
 \mathrm{TVD}\left(P^{(N)},P^{(n_1,n_2)}\right) = \frac{1}{2}\frac{1}{M} \sum_{\lambda}\left|1-\frac{L_{\lambda}(n_1,n_2)}{\overline{M}}\right|.
\end{equation}
For fixed $d$, we also have
\begin{equation}
M \sim \frac{N^{d-1}}{(d-1)!}+O\left(N^{d-2}\right), \quad M_i \sim \frac{n_i^{d-1}}{(d-1)!}+O\left(N^{d-2}\right) \quad \text{for} \quad i=1,2.
\end{equation}
Hence
\begin{equation}
\overline{M} =\frac{M_1M_2}{M}=\frac{n_1^{d-1}}{(d-1)!}(1-\alpha)^{d-1}(1+o(1)).
\end{equation}
Since $d$ is fixed, the inclusion-exclusion formula contains finitely many terms
only. Moreover,
\begin{equation}
    0\le\left(1-\sum_{j\in S}\frac{\ell_j}{\alpha}\right)_+^{d-1}\le 1,
\end{equation}
uniformly in $\ell\in\Delta_{d-1}$. Hence the remainder term in the asymptotic
expansion of $L_\lambda(n_1,n_2)$ is uniform over all histograms $\lambda$.
Combining the asymptotic expansion of $L_{\lambda}(n_1,n_2)$ from
Lemma~\ref{Postacpsi} with the asymptotics of $\overline{M}$, we obtain
\begin{equation}
\frac{L_{\lambda}(n_1,n_2)}{\overline{M}}=\frac{\varphi_\alpha(\ell)}{(1-\alpha)^{d-1}} + o(1),
\end{equation}
uniformly over all $\lambda$ such that $\ell=\frac{\lambda}{N} \in \Delta_{d-1}$.
Since the absolute value is a continuous function, it follows that
\begin{equation}
\left|1-\frac{L_{\lambda}(n_1,n_2)}{\overline{M}}\right|=\left|1 - \frac{\varphi_\alpha(\ell)}{(1-\alpha)^{d-1}}\right| + o(1).
\end{equation}
Therefore,
\begin{equation}
\frac{1}{M} \sum_{\lambda}\left|1-\frac{L_{\lambda}(n_1,n_2)}{\overline{M}}\right|=\frac{1}{M} \sum_{\lambda}\left|1 - \frac{\varphi_\alpha(\ell)}{(1-\alpha)^{d-1}}\right| + o(1).
\end{equation}
We now interpret the sum as a Riemann sum over the simplex. Histograms $\lambda$
satisfying $\sum_j\lambda_j=N$ correspond to lattice points in the dilated
simplex $N\cdot\Delta_{d-1}$. After normalization $\ell=\frac{\lambda}{N}$,
these points form a grid in $\Delta_{d-1}$ with mesh size $\frac1N$, which
becomes dense as $N\to\infty$. Since all histograms are equally weighted, the
normalized counting measure on the lattice points converges to the uniform
probability measure on $\Delta_{d-1}$.  Moreover, the function
\begin{equation}
\ell\mapsto\left|1-\frac{\varphi_\alpha(\ell)}{(1-\alpha)^{d-1}}\right|
\end{equation}
is bounded and piecewise continuous on the compact simplex $\Delta_{d-1}$.
Therefore, by the standard Riemann-sum approximation theorem,
\begin{equation}
\frac1M\sum_\lambda\left|1-\frac{\varphi_\alpha(\ell)}{(1-\alpha)^{d-1}}\right|\longrightarrow\mathbb E_{\ell\sim\mathrm{Unif}(\Delta_{d-1})}\left[\left|1-\frac{\varphi_\alpha(\ell)}{(1-\alpha)^{d-1}}\right|\right].
\end{equation}
Hence
\begin{equation}
\mathrm{TVD}\left(P^{(N)},P^{(n_1,n_2)}\right)=\frac12\mathbb E_{\ell\sim \mathrm{Unif}(\Delta_{d-1})}\left[\left|1-\frac{\varphi_\alpha(\ell)}{(1-\alpha)^{d-1}}\right|\right]+o(1).
\end{equation}
\end{proof}
\begin{uw}
    The model is invariant under exchanging the two blocks of sizes $n_1$ and
    $n_2$. Hence the aggregate total variation distance depends only on the
    unordered pair $\{\alpha,1-\alpha\}$:
    \begin{equation}
        \mathrm{TVD}_{\infty}(d,\alpha)=\mathrm{TVD}_{\infty}(d,1-\alpha)
    \end{equation}
\end{uw}
\subsubsection{Low-dimensional examples}

\paragraph{Case $d=2$}
\begin{tw}
    Fix $\alpha\in(0,1)$. In the fixed-$d$, large-$N$ regime with $d=2$, a
    normalized histogram can be parametrized as $\ell=(t,1-t)$ with $t \in [0,1]$.
    Then
    \begin{equation}
    \mathrm{TVD}_\infty(2,\alpha)=\alpha(1-\alpha)
    \end{equation}
\end{tw}
\begin{proof}
    In this case we obtain
\begin{equation}
\varphi_{2,\alpha}(t)=\max \left\{0,\min \left(1,\frac{t}{\alpha} \right)+\min \left(1,\frac{1-t}{\alpha} \right)-1 \right\}.
\end{equation}
Plugging this into the limiting total variation expression yields the
one-dimensional integral
\begin{equation}
\mathrm{TVD}_\infty(2,\alpha)=\frac12\int_0^1\left|1-\frac{\varphi_{2,\alpha}(t)}{1-\alpha}\right|\,dt.
\end{equation}
By symmetry, assume $\alpha \le \frac12$. Then
\begin{equation}
\varphi_{2,\alpha}(t)=
\begin{cases}
\frac{t}{\alpha}, & 0 \le t \le \alpha, \\
1, & \alpha \le t \le 1-\alpha, \\
\frac{1-t}{\alpha}, & 1-\alpha \le t \le 1.
\end{cases}
\end{equation}
Splitting the integral over the three intervals gives
\begin{equation}
\begin{split}
&\mathrm{TVD}_\infty(2,\alpha) =\frac12 \left(\int_0^\alpha \left|1 - \frac{t}{\alpha(1-\alpha)}\right| dt+ \int_\alpha^{1-\alpha} \frac{\alpha}{1-\alpha} \, dt+ \int_{1-\alpha}^1 \left|1 - \frac{1-t}{\alpha(1-\alpha)}\right| dt\right) \\
&= \alpha(1-\alpha).
\end{split}
\end{equation}
\end{proof}
To validate this expression, we compare it to exact finite-$N$ computations for
several values of $\alpha$. Numerically one observes the convergence
\begin{itemize}
\item $\alpha=0.5$: $N=5000$ gives $\mathrm{TVD}\approx 0.2498 \to 0.25 =
0.5\cdot 0.5$.
\item $\alpha=0.25$: $N=5000$ gives $\mathrm{TVD}\approx 0.1874 \to 0.1875 =
0.25\cdot 0.75$.
\item $\alpha=0.1$: $N=5000$ gives $\mathrm{TVD}\approx 0.09107 \to 0.09 =
0.1\cdot 0.9$.
\end{itemize}
This confirms consistency between the exact discrete model and the asymptotic
analysis.

\paragraph{Case $d=3$ (triangle)}

\begin{propo}
    For $d=3$ we parameterize a normalized histogram by $\ell=(x,y,z)$ with
    $x,y,z \ge 0$ and $x+y+z=1$. Then
    \begin{equation}
        \mathrm{TVD}_\infty(3,\alpha) = \frac{2}{(1-\alpha)^2} \int_{\Delta_2} \!\bigl(\varphi_{3,\alpha}(x,y)-(1-\alpha)^2\bigr)_+\,dx\,dy.
    \end{equation}
\end{propo}
\begin{proof}
    Using coordinates $(x,y)$ on the standard triangle $\Delta_2 = \{(x,y): x\ge
    0,\, y\ge 0,\, x+y \le 1\}$ and setting $z = 1 - x - y$, the appearance of
    the $(\cdot)_+$ operator follows from the identity $1- \min(1,u) = (1-u)_+$,
    so the truncation implemented by the $\min$ in the Ehrhart counting formula
    becomes the positive part in the piecewise-polynomial limit. This allows us
    to write all terms in inclusion–exclusion in a uniform form
\begin{equation}
\varphi_{3,\alpha}(x,y,z)= \sum_{S\subseteq\{1,2,3\}} (-1)^{|S|} \Bigl(1 - \tfrac{1}{\alpha}\!\sum_{j\in S}\ell_j\Bigr)^2_+.
\end{equation}
Formally the sum runs over all subsets $S\subseteq\{1,2,3\}$, including
$S=\{1,2,3\}$. The three-element term equals $\bigl(1-
\frac{1}{\alpha}\bigr)^2_+=0$ for all $\alpha\in(0,1)$ because $\frac1\alpha>1$,
hence it is omitted in the compact form. Explicitly,
\begin{equation}
    \begin{split}
        &\varphi_{3,\alpha}(x,y,z)=1-\sum_{j=1}^3 \Bigl(1-\frac{\ell_j}{\alpha}\Bigr)_+^2+\sum_{1\le i<j\le3}\Bigl(1-\frac{\ell_i+\ell_j}{\alpha}\Bigr)_+^2 \\
&=1-\Bigl(1-\frac{x}{\alpha}\Bigr)_+^2-\Bigl(1-\frac{y}{\alpha}\Bigr)_+^2-\Bigl(1-\frac{1-x-y}{\alpha}\Bigr)_+^2 + \\& + \Bigl(1-\frac{x+y}{\alpha}\Bigr)_+^2+\Bigl(1-\frac{1-y}{\alpha}\Bigr)_+^2+\Bigl(1-\frac{1-x}{\alpha}\Bigr)_+^2.
    \end{split}
\end{equation}
Since each $(\cdot)_+$ switches between zero and a linear expression at the
thresholds $\ell_j=\alpha$ and $\ell_i+\ell_j=\alpha$, $\varphi_{3,\alpha}$ is a
piecewise quadratic function on $\Delta_2$.

The limiting aggregate total variation distance can be expressed as
\begin{equation}
\mathrm{TVD}_\infty(3,\alpha)=\frac12\,\mathbb{E}_{\ell\sim\mathrm{Unif}(\Delta_{2})}\!\Bigl[\bigl|1 - \tfrac{\varphi_{3,\alpha}(\ell)}{(1-\alpha)^{2}}\bigr|\Bigr]=\frac12 \int_{\Delta_2}\left|1-\frac{\varphi_{3,\alpha}(x,y)}{(1-\alpha)^2}\right|\,2\,dx\,dy.
\end{equation}
Since $\mathbb{E}[\varphi_{3,\alpha}]=(1-\alpha)^2$, we can rewrite the absolute
deviation as
\begin{equation}
\left|1 - \frac{\varphi_{3,\alpha}}{(1-\alpha)^2}\right|=\frac{\bigl|(1-\alpha)^2 - \varphi_{3,\alpha}\bigr|}{(1-\alpha)^2}.
\end{equation}
For any random variable $X$ with $\mathbb{E}[X]=m$ we have the identity
\begin{equation}
\mathbb{E}[|X-m|] = 2\,\mathbb{E}[(X-m)_+],
\end{equation}
because positive and negative deviations from the mean have equal expectation.
Applying this with $X=\varphi_{3,\alpha}$ and $m=(1-\alpha)^2$ gives
\begin{equation}
\mathrm{TVD}_\infty(3,\alpha)=\frac{1}{(1-\alpha)^2}\,\mathbb{E}\!\bigl[(\varphi_{3,\alpha}-(1-\alpha)^2)_+\bigr].
\end{equation}
Finally, using that the density of $\mathrm{Unif}(\Delta_2)$ is $2$ on
$\Delta_2$ (whose area is $\frac12$), we obtain
\begin{equation}
\mathrm{TVD}_\infty(3,\alpha)=\frac{2}{(1-\alpha)^2}\int_{\Delta_2}\!\bigl(\varphi_{3,\alpha}(x,y)-(1-\alpha)^2\bigr)_+\,dx\,dy.
\end{equation}
\end{proof}

\begin{figure}[H]
    \centering
\includegraphics[width=0.6\textwidth]{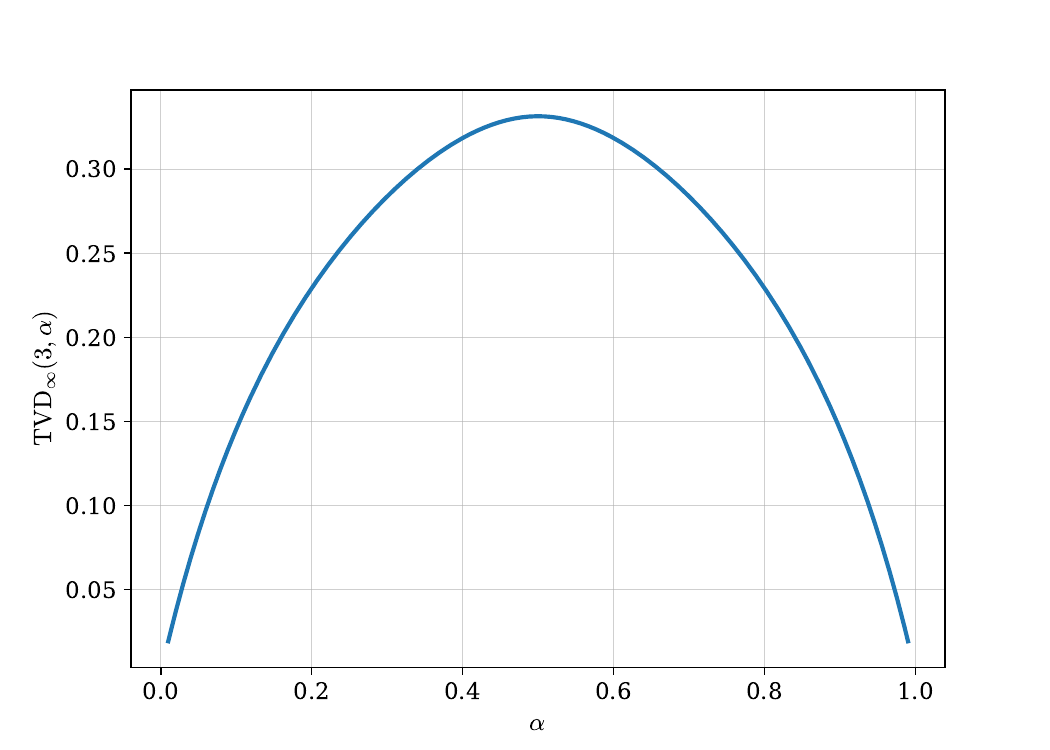}
\caption{Limiting aggregate total variation distance between $P^{(N)}(\lambda)$
and $P^{(n_1,n_2)}(\lambda)$, i.e. $\mathrm{TVD}_{\infty}(3,\alpha)$, as a
function of $\alpha$.}
    \label{d=3}
\end{figure}
The integral appearing in the definition of $\mathrm{TVD}_\infty(3,\alpha)$ does
not admit a simple closed-form expression because the number of polyhedral
regions grows rapidly. We illustrate this quantity numerically as a function of
$\alpha$. The integral was evaluated using a regular grid with $N=500$, and the
parameter $\alpha$ was sampled at $200$ equally spaced values. Numerically one
finds, for example,
\begin{equation*}
\mathrm{TVD}_\infty\!\left(3,\tfrac12\right)\approx 0.332,\qquad
\mathrm{TVD}_\infty\!\left(3,\tfrac13\right)\approx 0.298,\qquad
\mathrm{TVD}_\infty(3,0.1)\approx 0.145.
\end{equation*}

\paragraph{Case $d=4$ (tetrahedron)}
\begin{propo}
    For $d=4$ we write a normalized histogram as $\ell=(x,y,z,w)$ with
    $x,y,z,w\ge0$ and $x+y+z+w=1$. Then 
    \begin{equation}
        \mathrm{TVD}_\infty(4,\alpha)=\frac{6}{(1-\alpha)^3}\int_{\Delta_3}\big(\varphi_{4,\alpha}(x,y,z)-(1-\alpha)^3\big)_+\,dx\,dy\,dz.
\end{equation}
\end{propo}
\begin{proof}
Eliminating $w$ by $w=1-x-y-z$ we work on the standard $3$-simplex (tetrahedron)
\begin{equation}
\Delta_3=\{(x,y,z): x\ge0,\ y\ge0,\ z\ge0,\ x+y+z\le1\}.
\end{equation}
By the same inclusion–exclusion argument as before, the truncated-simplex volume
admits the uniform description
\begin{equation}
\varphi_{4,\alpha}(\ell) =\sum_{S\subseteq\{1,2,3,4\}}(-1)^{|S|} \Bigl(1-\frac{1}{\alpha}\sum_{j\in S}\ell_j\Bigr)_+^{3},
\end{equation}
i.e. each summand is \((1-\frac{1}{\alpha}\sum_{j\in S}\ell_j)_+^{d-1}\) with \(d-1=3\). The full-set term \(S=\{1,2,3,4\}\) equals \(\big(1-\frac{1}{\alpha}\big)^3_+=0\) for every \(\alpha\in(0,1)\) and is therefore omitted in the compact display. Explicitly the expansion contains singletons, pairs, triples and the empty set (the leading \(1\)). Then
\begin{equation}
    \begin{split}
        &\varphi_{4,\alpha}(x,y,z,w)= 1 -\sum_{j=1}^4\Bigl(1-\frac{\ell_j}{\alpha}\Bigr)_+^3 +\sum_{1\le i<j\le4}\Bigl(1-\frac{\ell_i+\ell_j}{\alpha}\Bigr)_+^3 -\sum_{1\le i<j<k\le4}\Bigl(1-\frac{\ell_i+\ell_j+\ell_k}{\alpha}\Bigr)_+^3\\
    \end{split}
\end{equation}

The limiting aggregate total variation distance is given by
\begin{equation}
\begin{split}
    &\mathrm{TVD}_\infty(4,\alpha)=\frac12\,\mathbb{E}_{\ell\sim\mathrm{Unif}(\Delta_3)}\!\Bigl[\Big|1-\frac{\varphi_{4,\alpha}(\ell)}{(1-\alpha)^3}\Big|\Bigr] = \frac{1}{2(1-\alpha)^3}\mathbb{E}_{\ell\sim\mathrm{Unif}(\Delta_3)}\left[\big|(1-\alpha)^3-\varphi_{4,\alpha}(\ell) \big|\right]\\
    &=\frac{1}{(1-\alpha)^3}\mathbb{E}_{\ell\sim\mathrm{Unif}(\Delta_3)}\!\left[\big(\varphi_{4,\alpha}(\ell)-(1-\alpha)^3\big)_+\right].
\end{split}
\end{equation}
Equivalently (writing the uniform expectation as an integral over the
tetrahedron), since the uniform density on \(\Delta_3\) is constant and equals
\(6\) (because $\operatorname{Vol}(\Delta_3)=\frac{1}{6}$), we have
\begin{equation}
\mathrm{TVD}_\infty(4,\alpha)=\frac12\int_{\Delta_3}\left|1-\frac{\varphi_{4,\alpha}(x,y,z)}{(1-\alpha)^3}\right|\,6\,dx\,dy\,dz.
\end{equation}
Replacing the expectation by the integral over \(T_3\) with density \(6\) gives
the computationally useful form
\begin{equation}
\mathrm{TVD}_\infty(4,\alpha)=\frac{6}{(1-\alpha)^3}\int_{\Delta_3}\big(\varphi_{4,\alpha}(x,y,z)-(1-\alpha)^3\big)_+\,dx\,dy\,dz.
\end{equation}
\end{proof}
Since the integral appearing in the definition of
$\mathrm{TVD}_\infty(4,\alpha)$ does not admit a simple closed-form expression,
we compute this quantity numerically as a function of $\alpha$. The integral
over the three-dimensional simplex was approximated using a Monte Carlo
integration scheme, and the parameter $\alpha$ was sampled at $100$ equally
spaced values.
\begin{figure}[H]
    \centering
    \includegraphics[width=0.5\textwidth]{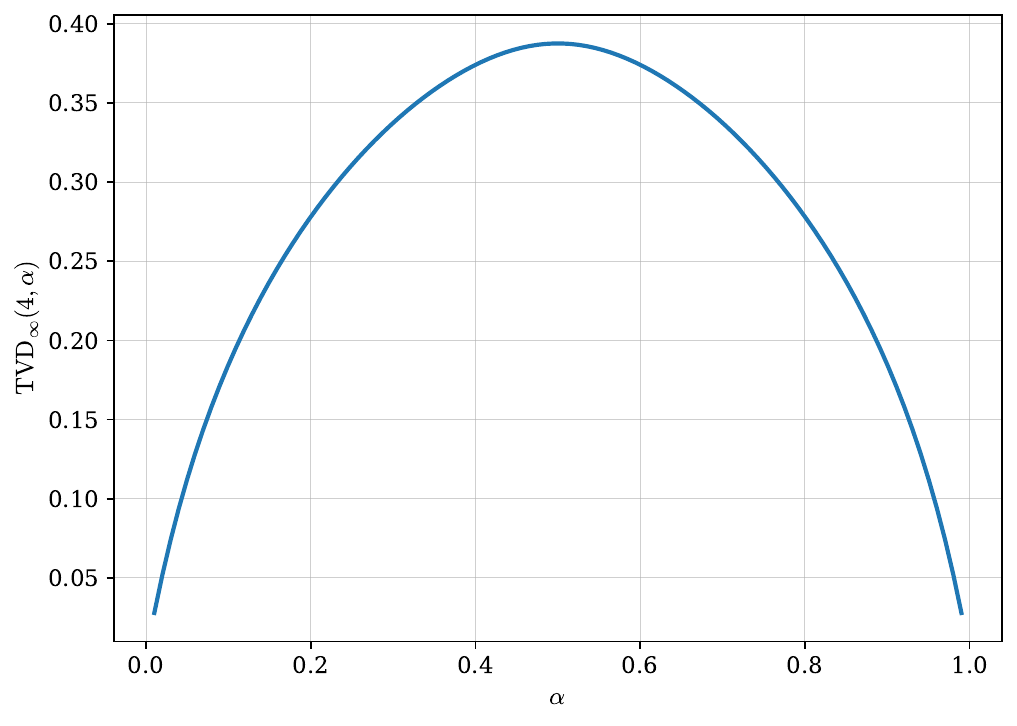}
    \caption{Limiting aggregate total variation distance between
    $P^{(N)}(\lambda)$ and $P^{(n_1,n_2)}(\lambda)$, i.e.
    $\mathrm{TVD}_{\infty}(4,\alpha)$, as a function of $\alpha$.}
    \label{d=4}
\end{figure}
Numerically:
\begin{equation*}
\mathrm{TVD}_\infty\!\left(4,\tfrac12\right)\approx 0.388,\qquad
\mathrm{TVD}_\infty\!\left(4,\tfrac13\right)\approx 0.352,\qquad
\mathrm{TVD}_\infty(4,0.1)\approx 0.184.
\end{equation*}

\subsubsection{Summary}
In summary, for $d=2,3,4$ we have derived a closed-form expression in the case
$d=2$ and integral representations together with numerical evaluations for
$d=3,4$.  The piecewise-polynomial structure of $\varphi_{d,\alpha}$ leads to
increasingly complex integrals as $d$ grows. Our numerical results illustrate
how $\mathrm{TVD}_\infty(d,\alpha)$ depends on the parameter $\alpha$. We also
observe that the plots of the aggregate total variation distance as a function
of $\alpha$ are symmetric.

\subsubsection{\texorpdfstring{The symmetric split $\alpha=\frac{1}{2}$}{The
symmetric split alpha=1/2}}

We now focus on the case
\begin{equation}
n_1=n_2=\frac{N}{2},
\end{equation}
which plays a distinguished role in the fixed-$d$, large-$N$ regime. In this
case the two blocks have equal size, and the limiting aggregate total variation
distance exhibits additional symmetry and structural simplifications. Recall
that for fixed $d$ the limiting aggregate total variation distance is given by
\begin{equation}
\mathrm{TVD}_\infty(d,\alpha)=\frac12\,\mathbb E_{\ell\sim\mathrm{Unif}(\Delta_{d-1})}\left|1-\frac{\varphi_\alpha(\ell)}{(1-\alpha)^{d-1}}\right|.
\end{equation}
Setting $\alpha=\frac{1}{2}$ yields the universal representation
\begin{equation}
    \mathrm{TVD}_\infty\!\left(d,\tfrac12\right) =\frac12\, \mathbb E_{\ell\sim\mathrm{Unif}(\Delta_{d-1})} \left| 1-2^{d-1}\, \varphi_{\frac{1}{2}}(\ell) \right|,
\end{equation}
where
\begin{equation}
    \varphi_{\frac{1}{2}}(\ell)=\sum_{S\subseteq[d]}(-1)^{|S|} \Bigl(1-\!\sum_{j\in S} \min\{1,2\ell_j\}\Bigr)_+^{\,d-1}= \sum_{S\subseteq[d]}(-1)^{|S|} \Bigl(1-2\!\sum_{j\in S}\ell_j\Bigr)_+^{\,d-1}.
\end{equation}
Here the second equality follows from the fact that whenever $\ell_j>\tfrac12$
for some $j\in S$, the corresponding term vanishes due to the positive-part
truncation. We next observe that the mean value of $\varphi_{\frac12}$ can be
computed explicitly. Indeed, for general $\alpha$ one has
\begin{equation}
\mathbb E_{\ell\sim\mathrm{Unif}(\Delta_{d-1})}[\varphi_\alpha(\ell)]=(1-\alpha)^{d-1},
\end{equation}
which follows from the normalization of the truncated simplex volume.
Specializing to $\alpha=\tfrac12$ yields
\begin{equation}
\mathbb E_{\ell\sim\mathrm{Unif}(\Delta_{d-1})}\bigl[\varphi_{\frac12}(\ell)\bigr]=2^{-(d-1)}.
\end{equation}
Using the identity
\begin{equation}
\mathbb E\bigl[\,|X-\mathbb E X|\,\bigr]=2\,\mathbb E\bigl[(X-\mathbb E X)_+\bigr],
\end{equation}
valid for any integrable random variable $X$, together with
$\mathbb E[\varphi_{\frac12}(\ell)]=2^{-(d-1)}$, we obtain
\begin{equation}\label{TVDpositive}
    \begin{split}
    &\mathrm{TVD}_\infty\!\left(d,\tfrac12\right) = \frac{1}{2}2^{d-1} \mathbb{E}_{\ell\sim\mathrm{Unif}(\Delta_{d-1})} \left[ \left| \varphi_{\frac12}(\ell) - 2^{-(d-1)} \right|\right]\\ & = \frac{1}{2}2^{d-1} \mathbb{E}_{\ell\sim\mathrm{Unif}(\Delta_{d-1})}\left[ \left| \varphi_{\frac12}(\ell) -\mathbb{E}_{\ell\sim\mathrm{Unif}(\Delta_{d-1})}[\varphi_{\frac12}(\ell)] \right|\right] =2^{d-1} \mathbb{E}_{\ell\sim\mathrm{Unif}(\Delta_{d-1})} \left[ \bigl(\varphi_{\frac12}(\ell)-2^{-(d-1)}\bigr)_+ \right].\\
    \end{split}
\end{equation}
The values of $\mathrm{TVD}_{\infty}(d,\frac12)$ were computed using Monte Carlo
sampling from the Dirichlet distribution. For each dimension $d$, we generate
$n=10^5$ independent samples. The aggregate total variation distance was then
estimated as
\begin{equation}
2^{d-1} \mathbb{E}_{\ell\sim\mathrm{Unif}(\Delta_{d-1})} \left[ \bigl(\varphi_{\frac12}(\ell)-2^{-(d-1)}\bigr)_+ \right].
\end{equation}
Error bars represent $95\%$ confidence intervals, computed as twice the standard
error of the Monte Carlo estimator, 
\begin{equation}
\pm \frac{2 \sigma}{\left(\frac{1}{2}\right)^{d-1}\sqrt{n}} = \pm 2^d\frac{ \sigma}{\sqrt{n}},
\end{equation}
where $\sigma$ denotes the empirical standard deviation of the sample values.
\begin{figure}[H]
    \centering
\includegraphics[width=0.6\textwidth]{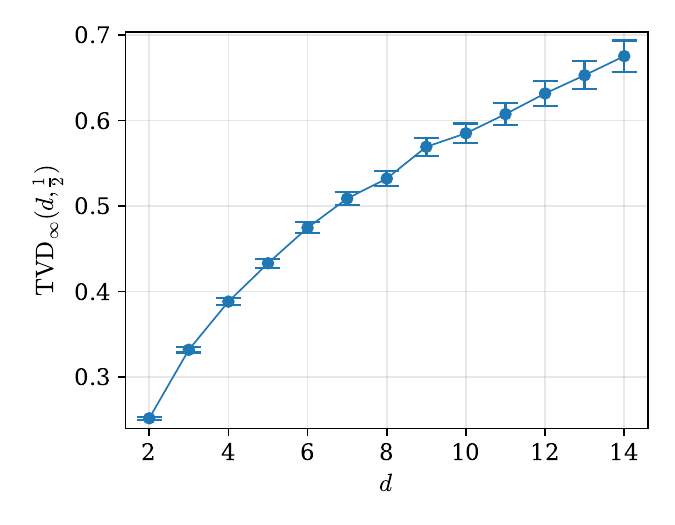}
\caption{The values of $\mathrm{TVD}_{\infty}(d,\frac12)$ with error bars.}
    \label{alpha12}
\end{figure}

\subsection{Joint scaling limits for aggregate model}

\subsubsection{\texorpdfstring{(JS-sparse) $N=o(\sqrt d)$}{(JS-sparse) N=o(sqrt d)}}
In the sparse joint-scaling regime $N=o(\sqrt d)$, the total number of
measurement outcomes grows sufficiently slowly compared to the local dimension
$d$ so that collisions of outcomes remain rare. As in the fixed-$N$, large-$d$
analysis, the dominant contribution to the aggregate total variation distance
arises from histograms with at most one collision. This allows us to reduce the
problem to a finite stratification of the histogram space and to exploit the
combinatorial structure developed in Section~\ref{collision}. More precisely, we
partition the set of histograms according to their collision structure:
\begin{itemize}
  \item $S_0$: no collisions (all counts $0$ or $1$),
  \item $S_1$: exactly one doubleton.
\end{itemize}
\begin{propo}[Sparse joint scaling]\label{JSsparse}
    Let $N=N(d)$ satisfy $N=o(\sqrt{d})$, and let
    $n_1$, $n_2$ be integers with $n_1+n_2=N$. Then, as
    $d \to \infty$, we have,
    \begin{equation}
        \mathrm{TVD}\left(P^{(N)},P^{(n_1,n_2)}\right)= \frac{n_1n_2}{d}+O\!\left(\frac{N^4}{d^2}\right).
    \end{equation}
\end{propo}
\begin{proof}
The argument follows the same collision-based stratification as in
Proposition~\ref{fixedN}. In the regime $N=o(\sqrt{d})$, the contribution of
configurations with two or more collisions is negligible, and the aggregate
$\mathrm{TVD}$ is captured by the strata $S_0$ and $S_1$. The leading
asymptotics are obtained by combining explicit formulas for
$L_{\lambda}(n_1,n_2)$ with expansions of binomial coefficients and of the
normalization factor $\overline{M}$. Full details are provided in
Appendix~\ref{sparseproof}.
\end{proof}
\begin{uw}
    In particular, since $N=o(\sqrt{d})$, we have 
    \begin{equation}
       \frac{N^4}{d^2}=o(1),
    \end{equation}
    where the rate of convergence depends on the scaling of $N$ relative to $d$.
    If $n_1/N\to\alpha$, we obtain
    \begin{equation}
        \mathrm{TVD}\left(P^{(N)},P^{(n_1,n_2)}\right) = \frac{\alpha(1-\alpha)N^2}{d}+o(1).
    \end{equation}
\end{uw}
\begin{cor}
    Under the assumptions of Proposition~\ref{JSsparse}, we have
    \begin{equation}
    \mathrm{TVD}\left(P^{(N)},P^{(n_1,n_2)}\right) \to 0 \quad \text{as} \quad d \to \infty
    \end{equation}
\end{cor}
This shows that in the sparse joint-scaling regime, the aggregate total
variation distance is asymptotically governed by single-collision events, while
configurations with multiple collisions contribute only at lower order.

\subsubsection{\texorpdfstring{Critical scaling regime $N/\sqrt d\to c$}{Critical
scaling regime $N/ \sqrt{d} \to c$}}

We now investigate the critical joint-scaling regime in which the number of
measurement outcomes grows proportionally to the square root of the local
Hilbert space dimension, $N/\sqrt d\to c$, with a fixed constant $c>0$. As
before, we consider a block decomposition with $n_1/N\to\alpha$ and
$n_2/N\to1-\alpha$, where $\alpha \in (0,1)$ is independent of $d$. Let
$\lambda=(\lambda_0,\dots,\lambda_{d-1})$ be a random histogram of size $N$,
drawn uniformly from all weak compositions of $N$ into $d$ parts. Define the
collision count by
\begin{equation}
C(\lambda) := \sum_{i=1}^d \binom{\lambda_i}{2}.
\end{equation}
\begin{propo}[Poisson limit for the collision count]\label{Poisson}
    Let $N$ be an integer sequence such that $N/\sqrt d\to c$ for some
    $c>0$. Then the collision count $C(\lambda)$ converges in
    distribution to a Poisson random variable $K$ with mean $c^2$:
    \begin{equation}
        C(\lambda)\xrightarrow[d\to \infty]{\mathrm{d}} K,\qquad K\sim \mathrm{Poisson}(c^2).
   \end{equation}
   Equivalently, for each fixed $k\ge 0$,
   \begin{equation}
       \mathbb P(C(\lambda)=k)\longrightarrow e^{-c^2}\frac{(c^2)^k}{k!}.
   \end{equation}
\end{propo}
\begin{proof}
    This type of Poisson approximation is closely related to classical occupancy
    problems and Chen-Stein techniques, see \cite{arratia1989two}. The proof is
    given in Appendix~\ref{lkolizji}.
\end{proof}
\begin{propo}[Decomposition count for finitely many doubletons]\label{grLn}
    Fix $k\ge0$. Let $\lambda$ be a histogram containing $k$ bins of
    multiplicity $2$ and no higher multiplicities. If $n_1/N\to\alpha\in(0,1)$,
    then
    \begin{equation}
    \frac{L_{\lambda}(n_1,n_2)}{\binom{N}{n_1}} \longrightarrow (1-\alpha(1-\alpha))^k.
    \end{equation}
\end{propo}
\begin{proof}
    See Appendix~\ref{L/binom}.
\end{proof}
\begin{lem}[Critical normalization]\label{normalizationlem}
    Let $\overline{M}=\frac{M_1M_2}{M}$. If $N=N(d)$ satisfies
    $N/\sqrt d\to c>0$ and $n_1/N\to\alpha\in(0,1)$, then
    \begin{equation}
    \frac{\binom{N}{n_1}}{\overline{M}} \sim e^{\alpha(1-\alpha)c^2}.
    \end{equation}
\end{lem}
\begin{proof}
    The derivation is given in Appendix~\ref{normalization}.
\end{proof}
\begin{lem}[Critical aggregate domination]\label{criticaldomination}
Under the assumptions of Lemma~\ref{normalizationlem}, the likelihood ratios
\begin{equation}
R_d(\lambda)=\frac{L_\lambda(n_1,n_2)}{\overline M}
\end{equation}
are uniformly bounded for all sufficiently large $d$. Consequently,
$\{|1-R_d(\lambda)|\}_d$ is uniformly integrable under
$\lambda\sim P^{(N)}$.
\end{lem}
\begin{proof}
Every admissible decomposition of a histogram is determined by choosing which
$n_1$ of the $N$ occurrences belong to the first block. Hence
$L_\lambda(n_1,n_2)\le\binom{N}{n_1}$ for every $\lambda$. Lemma~\ref{normalizationlem}
shows that $\binom{N}{n_1}/\overline M$ converges to
$e^{\alpha(1-\alpha)c^2}$, and is therefore bounded for large $d$.
\end{proof}
Intuitively, in the critical regime the histogram consists mostly of singletons
together with a Poisson number of double collisions.  Each collision modifies
the number of admissible histogram decompositions by a factor
$(1-\alpha(1-\alpha))$, leading to the expectation formula above.
\begin{tw}\label{twCritScal}
    Let $N$ satisfy $N/\sqrt d\to c>0$, and let
    $n_1/N\to\alpha\in(0,1)$ with $n_2=N-n_1$. Then the aggregate total variation
    distance converges to
    \begin{equation}
    \lim_{d \to \infty} \mathrm{TVD}\left(P^{(N)},P^{(n_1,n_2)}\right)=\frac{1}{2}\mathbb{E} \left| 1 - e^{\alpha(1-\alpha) c^2}(1-\alpha(1-\alpha))^K \right|,
    \end{equation}
    where 
    \begin{equation}
    K \sim \mathrm{Poisson}(c^2).
    \end{equation}
\end{tw}
\begin{proof}
    Combining Propositions~\ref{Poisson}-\ref{grLn} with
    Lemma~\ref{normalizationlem} reduces the aggregate $\mathrm{TVD}$ expression
    to an expectation over the collision count $K$. The details are given in
    Appendix~\ref{ProofCritical}.
\end{proof}

\subsection{Block-resolved asymptotics}
In this subsection we briefly compare the asymptotic behaviour of the
block-resolved total variation distance with the aggregate statistic analyzed
above. As shown in Section~\ref{blockresolved}, the block-resolved experiment
retains additional information about the location of collisions across the two
blocks. Consequently, the corresponding $\mathrm{TVD}$ is generally larger than
the aggregate $\mathrm{TVD}$.
\begin{propo}[Sparse block-resolved asymptotics]\label{SparseResolved}
The leading sparse signal is unchanged. Fix integers $N\ge1$ and
$n_1,n_2\ge0$ with $n_1+n_2=N$. As $d\to\infty$,
\begin{equation}
\mathrm{TVD}\left(P_{\mathrm{sh}},P_{\mathrm{ind}}\right)=\frac{n_1n_2}{d}+O(d^{-2}),
\end{equation}
with constants depending only on $N$. More generally, if $N=o(\sqrt d)$ and
$n_1$, $n_2=N-n_1$, then, uniformly over the block split,
\begin{equation}
\mathrm{TVD}\left(P_{\mathrm{sh}},P_{\mathrm{ind}}\right)=\frac{n_1n_2}{d}+O\!\left(\frac{N^4}{d^2}\right).
\end{equation}
\end{propo}
The argument follows the same collision-sector decomposition as in the aggregate case. The leading contribution comes from simple cross-block collisions. Full details are provided in Appendix~\ref{ProofSR}.

\begin{tw} \label{CriticalResolved}
Let $N$ satisfy $N/\sqrt d\to c>0$, and let
$n_1/N\to\alpha \in (0,1)$ with $n_2=N-n_1$. The block-resolved limit is governed
by the number of simple cross-block collisions.
With
\begin{equation}
\beta:=\alpha(1-\alpha)c^2,\qquad C\sim\mathrm{Poisson}(\beta),
\end{equation}
one obtains
\begin{equation}
\mathrm{TVD}\left(P_{\mathrm{sh}},P_{\mathrm{ind}}\right)\longrightarrow\frac12\,\mathbb E\left[\left|e^{-\beta}2^C-1\right|\right].
\end{equation}
\end{tw}
In the critical regime, the number of simple cross-block collisions converges to a Poisson random variable. Each such collision contributes a factor of two to the likelihood ratio, while all higher-order collision events are asymptotically negligible. Combining the Poisson approximation with the asymptotics of the normalization factor yields the stated limit. Full details are given in Appendix~\ref{ProofCritResolved}.

\begin{cor} \label{ResolvedFixedd}
For fixed $d\ge2$ and $N\to\infty$, if $\frac{n_1}{N}\to\alpha\in(0,1)$ and
$n_2=N-n_1$, then
\begin{equation}
\mathrm{TVD}\left(P_{\mathrm{sh}},P_{\mathrm{ind}}\right)\longrightarrow 1.
\end{equation}
\end{cor}
The shared and independent models converge to mutually singular limiting distributions on the simplex product space. The claim then follows from standard properties of total variation distance under weak convergence. See Appendix~\ref{ProofResolvedFixedd} for details. In the block-resolved experiment the two models become asymptotically perfectly distinguishable.

\section{Conclusion} 
In this work we studied distinguishability questions generated by Haar-random
measurement channels at two levels of observation. At the channel level, we
considered a maximally entangled tester for distinguishing a Haar-random
measure-and-prepare channel from the identity channel. This provides a
coherence-sensitive benchmark for the random measurement family. At the
classical level, we then asked what can be distinguished after the measurement
has already produced outcome data. In that setting, the problem becomes a
comparison of occupancy statistics and, in particular, of collision structures
in random histograms.

For the comparison of the two Haar-averaged measurement models, we obtained
exact formulas for the block-unresolved aggregate histogram laws, proved the
large-$d$ asymptotic
\begin{equation}
\mathrm{TVD}\left(P^{(N)},P^{(n_1,n_2)}\right)=\frac{n_1n_2}{d}+O(d^{-2})
\end{equation}
for fixed $N$, and derived simplex-limit formulas for fixed $d$ and large $N$.
We also showed that the same leading collision-driven asymptotic persists in the
sparse joint-scaling regime $N=o(\sqrt{d})$.

We also clarified the operational role of the statistic being observed. If block
labels are retained, the ordered pair of block histograms has an exact
likelihood ratio governed by
$\prod_j\binom{\lambda^{(1)}_j+\lambda^{(2)}_j}{\lambda^{(1)}_j}$, and its TVD
dominates the aggregate TVD by data processing. In sparse large-$d$ regimes the
first-order asymptotics agree with the aggregate result, but for fixed $d$ and
large $N$ the block-resolved TVD tends to one. Thus the nontrivial simplex
limits obtained in this manuscript characterize precisely the block-unresolved
experiment in which only the total histogram is available.

Our results illustrate how a quantum discrimination problem can be translated
into a hierarchy of statistical questions. Coherent testers detect the
measurement channel as a departure from the identity, while classical outcome
statistics reveal the structure of the underlying Haar randomness only through
histogram combinatorics. In the critical scaling regime $N/\sqrt d\to c$, we
established that the number of collision pairs converges to a Poisson
distribution with mean $c^2$. This leads to an explicit asymptotic formula for
the aggregate total variation distance between the collective and
block-separated models.   

\section*{Acknowledgements}
MM and ŁP acknowledges support from the National Science Center (NCN), Poland, under
Project Opus No. 2024/53/B/ST2/02026. 
ZP and LM acknowledge support from the
National Science Center (NCN), Poland, under Project Opus No. 
2022/47/B/ST6/02380
\bibliographystyle{quantum}
\bibliography{bibliografia}

\appendix
\section{Weingarten calculus} \label{WeingartenCal}

\subsection{\texorpdfstring{Evaluation of the integral $I_{n,d}$}{Evaluation of
the integral I(n,d)}}\label{ObInd}

We evaluate
\begin{equation}
   \frac{1}{d^{2N}} I_{N,d}=\frac{1}{d^{2N}}\sum_{\mathbf{k}}\int\big|\bra{\mathbf{k}}U^{\otimes N}\ket{\mathbf{k}}\big|^2\,dU, 
\end{equation}
where the integration is with respect to the Haar measure on $U(d)$. For a fixed
multi-index $\mathbf{k}=(k_0,\dots,k_{N-1})$ we have
\begin{equation}
\bra{\mathbf{k}}U^{\otimes N}\ket{\mathbf{k}}=\prod_{r=0}^{N-1} U_{k_r k_r},\qquad\big|\bra{\mathbf{k}}U^{\otimes N}\ket{\mathbf{k}}\big|^2=\prod_{r=0}^{N-1} |U_{k_r k_r}|^2.
\end{equation}
Summing over all $\mathbf{k}\in[d]^N$ yields
\begin{equation}
\sum_{\mathbf{k}}\prod_{r=0}^{N-1} |U_{k_r k_r}|^2=\Big(\sum_{j=0}^{d-1} |U_{jj}|^2\Big)^N.
\end{equation}
Thus we obtain
\begin{equation}
I_{N,d}=\int_{U(d)}\Big(\sum_{j=0}^{d-1} |U_{jj}|^2\Big)^N dU.
\end{equation}
To evaluate $I_{N,d}$ we use the general Weingarten formula
\cite{puchala2011symbolic,collins2006integration}. For arbitrary indices
$(i_r),(j_r),(i'_r),(j'_r)$ of length $N$ one has
\begin{equation}\label{weingarten}
\begin{split}
    &\int_{U(d)}\prod_{r=0}^{N-1} U_{i_r j_r}\,\overline{U_{i'_r j'_r}}\, dU=\sum_{\sigma,\tau\in S_N}\Big(\prod_{r=0}^{N-1}\delta_{i_r,\,i'_{\sigma(r)}}\Big)\Big(\prod_{r=0}^{N-1}\delta_{j_r,\,j'_{\tau(r)}}\Big)\mathrm{Wg}(\sigma^{-1}\tau,d).\\
\end{split}
\end{equation}
The following lemma provides a combinatorial interpretation of the delta
constraints appearing in the Weingarten formula and will be used in the proof of
the next result.
\begin{lem}\label{delta} For every pair of permutations $\sigma,\tau\in S_N$ one
has
\begin{equation}
\sum_{\mathbf{k}\in[d]^N}\prod_{r=0}^{N-1} \delta_{k_r,k_{\sigma(r)}} \,\delta_{k_r,k_{\tau(r)}}=d^{c(\sigma\vee\tau)},
\end{equation}
where $c(\sigma\vee\tau)$ denotes the number of connected components of the
graph obtained by joining $(r,\sigma(r))$ and $(r,\tau(r))$.
\end{lem}
\begin{proof}
Consider a graph $G$ with vertices $\{0,\dots,N-1\}$, in which for each $r$ we
add two edges:
\begin{equation}
r \sim \sigma(r), \qquad r \sim \tau(r).
\end{equation}
The product
\begin{equation}
\prod_{r=0}^{N-1} \delta_{k_r,k_{\sigma(r)}} \, \delta_{k_r,k_{\tau(r)}}
\end{equation}
equals $1$ if and only if all indices $k_r$ are equal within every connected
component of $G$, and is $0$ otherwise. If the graph $G$ has $c(\sigma\vee\tau)$
connected components, then in the tuple $\mathbf{k}=(k_0,\dots,k_{N-1})$ one
can freely choose $c(\sigma\vee\tau)$ values (one for each component), while the
remaining indices are determined. Since each value can be chosen in $d$ ways,
the total number of admissible $\mathbf{k}$ equals $d^{c(\sigma\vee\tau)}$.
\end{proof}
\begin{lem}
Specializing the Weingarten formula \eqref{weingarten} to the diagonal case
$i_r=j_r=i'_r=j'_r=k_r$ yields
\begin{equation}
I_{N,d}=\sum_{\sigma,\tau\in S_N}d^{c(\sigma\vee\tau)}\,\mathrm{Wg}(\sigma^{-1}\tau,d),
\end{equation}
where $c(\sigma\vee\tau)$ denotes the number of blocks in the join of the cycle
partitions of $\sigma$ and $\tau$.
\end{lem}
\begin{proof}
        \begin{equation}
    \begin{split}
        & I_{N,d} = \int_{U(d)}\Big(\sum_{j=0}^{d-1} |U_{jj}|^2\Big)^N dU =\int_{U(d)} \sum_{\mathbf{k}} \prod_{r=0}^{N-1} |U_{k_r k_r}|^2 dU = \int_{U(d)} \sum_{\mathbf{k}}\prod_{r=0}^{N-1} U_{k_r k_r} \overline{U_{k_r k_r}} dU \\ 
        & = \sum_{\mathbf{k}} \sum_{\sigma,\tau\in S_N} \Big(\prod_{r=0}^{N-1}\delta_{k_r,k_{\sigma(r)}}\Big) \Big(\prod_{r=0}^{N-1}\delta_{k_r,k_{\tau(r)}}\Big)\mathrm{Wg}(\sigma^{-1}\tau,d)=\sum_{\sigma,\tau\in S_N} d^{c(\sigma\vee\tau)} \mathrm{Wg}(\sigma^{-1}\tau,d) \\
    \end{split}
\end{equation}
\end{proof}

\subsection{\texorpdfstring{Explicit evaluation for small $N$}{Explicit
evaluation for small N}}

\begin{prz}[Case $N=2$]
For $N=2$ the symmetric group $S_2$ has two elements: the identity $\mathrm{id}$
and the transposition $(12)$. Using the permutation-sum representation
\begin{equation}
I_{2,d}=\sum_{\sigma,\tau\in S_2} d^{c(\sigma\vee\tau)}\mathrm{Wg}(\sigma^{-1}\tau,d),
\end{equation}
we list the four contributions and their combinatorial factors:
\begin{equation}
\begin{array}{|c|c|c|}
\hline
(\sigma,\tau) & \sigma^{-1}\tau & c(\sigma\vee\tau) \\ \hline
(\mathrm{id},\mathrm{id}) & \mathrm{id} & 2 \\
(\mathrm{id},(12))       & (12)               & 1 \\
((12),\mathrm{id})       & (12)               & 1 \\
((12),(12))              & \mathrm{id}        & 1 \\
\hline
\end{array}
\end{equation}
Hence
\begin{equation}
I_{2,d} = (d^2+d)\,\mathrm{Wg}(\mathrm{id},d) + 2d\,\mathrm{Wg}((12),d).
\end{equation}
The Weingarten values for $N=2$  are
\begin{equation}
\mathrm{Wg}(\mathrm{id},d)=\frac{1}{d^2-1},
\end{equation}
\begin{equation}
\mathrm{Wg}((12),d)=-\frac{1}{d(d^2-1)}.
\end{equation}
Substituting these into the formula gives, for $d\ge2$,
\begin{equation}
I_{2,d} = \frac{d+2}{d+1}.
\end{equation}
\end{prz}

\begin{prz}[Case $N=3$]
For $N=3$ we start directly from the general formula
\begin{equation}
I_{3,d}=\sum_{\sigma,\tau\in S_3} d^{\,c(\sigma\vee\tau)}\,\mathrm{Wg}(\sigma^{-1}\tau,d).
\end{equation}
It is convenient to group the terms by the permutation $\pi=\sigma^{-1}\tau$.
For a fixed $\pi$ the pairs $(\sigma,\tau)$ with $\tau=\pi \sigma$ are in
one-to-one correspondence with $\sigma\in S_3$, so that
\begin{equation}
I_{3,d}=\sum_{\pi\in S_3} \mathrm{Wg}(\pi,d)\;\Bigg(\sum_{\sigma\in S_3} d^{c(\sigma\vee\sigma\pi)}\Bigg).
\end{equation}
There are three conjugacy classes in $S_3$:
\begin{itemize}
  \item $\pi=\mathrm{id}$,
  \item $\pi$ a transposition,
  \item $\pi$ a $3$-cycle.
\end{itemize}
We now compute the inner sums
\begin{equation}
S(\pi):=\sum_{\sigma\in S_3} d^{\,c(\sigma\vee\sigma\pi)}.
\end{equation}
\noindent\textbf{1.} For $\pi=\mathrm{id}$:  
If $\sigma=\mathrm{id}$ then $c(\sigma\vee\sigma\pi)=3$ giving $d^3$.  For the
three transpositions $c(\sigma\vee\sigma\pi)=2$, contributing $3d^2$.  For the
two $3$-cycles $c(\sigma\vee\sigma\pi)=1$, contributing $2d$.  Thus
\begin{equation}
S(\mathrm{id})=d^3+3d^2+2d.
\end{equation}
\noindent\textbf{2.} For $\pi$ a transposition:  
Among the six choices of $\sigma\in S_3$ we obtain two cases with
$c(\sigma\vee\sigma\pi)=2$ and four cases with $c(\sigma\vee\sigma\pi)=1$, hence
\begin{equation}
S(\text{transposition})=2d^2+4d.
\end{equation}
\noindent\textbf{3.} For $\pi$ a $3$-cycle:  
Here all six values of $\sigma$ yield $c(\sigma\vee\sigma\pi)=1$, so
\begin{equation}
S(\text{3-cycle})=6d.
\end{equation}
Collecting all contributions, we find
\begin{equation}
\begin{aligned}
I_{3,d}
&= \mathrm{Wg}(\mathrm{id},d)\,(d^3+3d^2+2d) + 3\,\mathrm{Wg}(\text{transposition},d)\,(2d^2+4d)  + 2\,\mathrm{Wg}(\text{3-cycle},d)\,(6d).
\end{aligned}
\end{equation}
\end{prz}
Using the standard tabulated Weingarten expressions, for $d\ge3$, we have
\begin{equation}
I_{3,d}=\frac{d^2+6d+10}{(d+1)(d+2)}.
\end{equation}
In the case of $d=2$, we have
\begin{equation}
    I_{3,2}=2.
\end{equation} 
Indeed, for $U\in U(2)$ one has
\begin{equation}
|U_{00}|^2+|U_{11}|^2 = 2|U_{00}|^2.
\end{equation}
Since $|U_{00}|^2$ is uniformly distributed on $[0,1]$, we obtain
\begin{equation}
I_{3,2}=\mathbb E[(2|U_{00}|^2)^3]=8 \mathbb{E}[(|U_{00}|^2)^3] =2.
\end{equation}

\subsection{\texorpdfstring{Asymptotic analysis for fixed $N$}{Asymptotic
analysis for fixed N}} \label{asymptotic} We now analyze the behavior of
$I_{N,d}$ for fixed $N$ and $d\to\infty$.  Recall that
\begin{equation}
I_{N,d}=\sum_{\sigma,\tau\in S_N} d^{c(\sigma\vee\tau)} \mathrm{Wg}(\sigma^{-1}\tau,d) =\sum_{\pi\in S_N}\mathrm{Wg}(\pi,d)\; S(\pi),\end{equation}
\begin{equation}
S(\pi):=\sum_{\sigma\in S_N} d^{\,c(\sigma\vee\sigma\pi)}.
\end{equation}
From the general asymptotics of the Weingarten function
\cite{collins2006integration} one has
\begin{equation}
\mathrm{Wg}(\pi,d)=d^{-N-|\pi|}\big(\mathrm{Moeb}(\pi)+O(d^{-2})\big),
\end{equation}
where $|\pi|=N-c(\pi)$ is the length of the permutation, and
$\mathrm{Moeb}(\pi)$ denotes the Möbius-type constant depending only on the
cycle type of $\pi$ (see \cite{collins2006integration} for the precise
definition). On the other hand, $S(\pi)$ is a polynomial in $d$ of degree
$N-|\pi|$. Hence the leading order of the product is
\begin{equation}
\mathrm{Wg}(\pi,d)\,S(\pi)\;\sim\; d^{-2|\pi|}.
\end{equation}
Therefore:
\begin{itemize}
    \item For $\pi=\mathrm{id}$ ($|\pi|=0$) one has
        \begin{equation}
            S(\mathrm{id})=\sum_{\sigma\in S_N} d^{c(\sigma)}=d^N+\binom{N}{2}d^{N-1}+O(d^{N-2}),
        \end{equation}
        and since
        \begin{equation}
            \mathrm{Wg}(\mathrm{id},d)=d^{-N}(1+O(d^{-2})),
        \end{equation}
        their product contributes at orders $1$, $\frac1d$,
        $\frac{1}{d^2}$.
  \item $\pi$ a transposition ($|\pi|=1$) contributes terms of order
  $\frac{1}{d^2}$,
  \item all other permutations ($|\pi|\ge 2$) contribute at most $O(d^{-4})$.
\end{itemize}
Thus the coefficient of $d^{-1}$ can only come from the identity permutation
$\pi=\mathrm{id}$; all nonidentity permutations contribute at order $O(d^{-2})$
or smaller.

\medskip

For $\pi=\mathrm{id}$ we have
\begin{equation}
    \begin{split}
        &S(\mathrm{id})=\sum_{\sigma\in S_N} d^{c(\sigma \vee \sigma)} =\sum_{\sigma\in S_N} d^{c(\sigma)}=  \sum_{k=1}^{N} \stirling{N}{k} d^k  = d(d+1)\cdots(d+N-1)\\
        &=d^N+\binom{N}{2}d^{N-1}+O(d^{N-2}).
    \end{split}
\end{equation}
where $\stirling{N}{k}$ is the (unsigned) Stirling number of the first kind.
These numbers count the permutations of $N$ elements that consist of exactly $k$
disjoint cycles, that is, $\stirling{N}{k}$ equals the number of $\sigma\in S_N$
such that $c(\sigma)=k$. Multiplying with
$\mathrm{Wg}(\mathrm{id},d)=d^{-N}(1+O(d^{-2}))$ yields
\begin{equation}
    \begin{split}
        & \mathrm{Wg}(\mathrm{id},d)\,S(\mathrm{id})= d^{-N}(1+O(d^{-2})) \left(d^N+\binom{N}{2}d^{N-1}+O(d^{N-2})\right) \\
        & = 1+\frac{\binom{N}{2}}{d}+O(d^{-2}).
    \end{split}
\end{equation}
For $\pi$ a transposition one has $|\pi|=1$, so $S(\pi)$ is of degree at most
$N-1$ and
\begin{equation}
\mathrm{Wg}(\pi,d)\,S(\pi)=O(d^{-2}).
\end{equation}
Since $N$ is fixed, summing over all nonidentity permutations contributes
$O(d^{-2})$ in total. Therefore
\begin{equation}
I_{N,d}=1+\frac{\binom{N}{2}}{d}+O(d^{-2}),
\end{equation}
so
\begin{equation}
    I_{N,\infty}=1.
\end{equation}
This completes the proof of the asymptotic expansion.

\section{Haar-averaged histogram distributions}
\subsection{Proof of Lemma \ref{permutacje}} \label{permutacjeobl} Consider
 the state $\rho = \ket{00\dots0}\!\bra{00\dots0}$ and, for simplicity, denote
 $N=n_1 + n_2$.
 \begin{equation}
     \begin{split}
         & \overline{Q}^{(N)}(\rho) = \sum_{\mathbf{k}} \text{tr}\left(\overline{E}_{\mathbf{k}}^{(N)}\ket{00\dots0}\!\bra{00\dots0} \right) \ket{\mathbf{k}}\!\bra{\mathbf{k}} \\
     \end{split}
 \end{equation}
 If $\mathbf{k} = \sigma(\mathbf{m})$, then
 \begin{equation}
     \begin{split}
         & \overline{E}_{\mathbf{k}}^{(N)} = \int U^{\dagger \otimes N} \ket{\mathbf{k}}\!\bra{\mathbf{k}} U^{ \otimes N} dU = \int U^{\dagger \otimes N} P_{\sigma} \ket{\mathbf{m}}\!\bra{\mathbf{m}} P_{\sigma}^T U^{\otimes N} dU\\
         & = \int P_{\sigma} P_{\sigma}^T U^{\dagger \otimes N} P_{\sigma} \ket{\mathbf{m}}\!\bra{\mathbf{m}} P_{\sigma}^T U^{ \otimes N}P_{\sigma} P_{\sigma}^T dU  = \int P_{\sigma} U^{\dagger\otimes N}  \ket{\mathbf{m}}\!\bra{\mathbf{m}} U^{ \otimes N}P_{\sigma}^T dU = P_{\sigma} \overline{E}_{\mathbf{m}}^{(N)} P_{\sigma}^T.
     \end{split}
 \end{equation}
Hence 
 \begin{equation}
     \begin{split}
         & \text{tr} \overline{E}_{\mathbf{k}}^{(N)} \ket{00\dots0}\!\bra{00\dots0} = \text{tr}P_{\sigma} \overline{E}_{\mathbf{m}}^{(N)} P_{\sigma}^T \ket{00\dots0}\!\bra{00\dots0} = \text{tr} \overline{E}_{\mathbf{m}}^{(N)} P_{\sigma}^T \ket{00\dots0}\!\bra{00\dots0} P_{\sigma} \\
         & = \text{tr} \overline{E}_{\mathbf{m}}^{(N)}\ket{00\dots0}\!\bra{00\dots0}.
     \end{split}
 \end{equation}
 That is, the trace depends only on the composition of $\mathbf{k}$, not on its
 particular ordering. On the other hand,
\begin{equation}
    \begin{split}
        & \text{tr} \overline{E}_{\mathbf{k}}^{(N)} \ket{00\dots0}\!\bra{00\dots0}  = \int \bra{00\dots0} U^{\dagger \otimes N} \ket{\mathbf{k}}\!\bra{\mathbf{k}} U^{\otimes N} \ket{00\dots0} dU \\
        & =  \int \bra{0}U^{\dagger}\ket{k_0}\dots\bra{0}U^{\dagger}\ket{k_{N-1}}\bra{k_0}U\ket{0}\dots\bra{k_{N-1}}U\ket{0} dU  = \int |\bra{k_0}U\ket{0}|^2 \dots |\bra{k_{N-1}}U\ket{0}|^2 dU \\
        & = \int \prod_{r=0}^{N-1} |\bra{k_r}U\ket{0}|^2 dU. \\
    \end{split}
\end{equation}
\subsection{Conditional multinomial distribution} \label{multinomial} For the
block-unresolved experiment, the variable $\Lambda$ provides a complete
description of the observed aggregate statistics. For the random channel
$Q_U^{(N)}$, the probability of obtaining a given outcome sequence depends only
on its histogram, not on the ordering of outcomes; for the product channel
$Q_{U,V}^{(n_1,n_2)}$, aggregate statistics are obtained by summing over all
block-wise histograms $\lambda^{(1)}$ and $\lambda^{(2)}$ satisfying
$\lambda^{(1)}+\lambda^{(2)}=\lambda.$ Denote $W_{k_r} :=
|\bra{k_r}U\ket{0}|^2$. Then 
\begin{equation}
\prod_{r=0}^{N-1} |\bra{k_r}U\ket{0}|^2 = \prod_{r=0}^{N-1} W_{k_r} = \prod_{j=0}^{d-1} W_j^{\lambda_j}.
\end{equation}
Hence we obtain the conditional multinomial distribution
\begin{equation}
    \begin{split}
        & \mathbb{P}(\Lambda=\lambda|U) =m(\lambda)(W_0)^{\lambda_0}(W_1)^{\lambda_1} \dots (W_{d-1})^{\lambda_{d-1}} = m(\lambda)\prod_{j=0}^{d-1} W_j^{\lambda_j}
    \end{split}
\end{equation}
where $m(\lambda)$ is the multinomial coefficient, i.e.,
\begin{equation}
m(\lambda) = \frac{N!}{\lambda_0!\lambda_1!\dots \lambda_{d-1}!}.\end{equation}

\subsection{Dirichlet integral and Haar averaging} \label{dirichlet}

The density of the Dirichlet distribution
$\text{Dir}(\alpha_0,\dots,\alpha_{d-1})$ on the simplex $\Delta_{d-1}$ is given
by
\begin{equation}
    f_{\text{Dir}(\alpha)}(p) = \frac{1}{B(\alpha)} \prod_{j=0}^{d-1} p_j^{\alpha_j -1}, \end{equation}
    where $\alpha=(\alpha_0,\dots,\alpha_{d-1}) $ is the vector of parameters and 
\begin{equation}  
B(\alpha) = \frac{\prod_{j=0}^{d-1} \Gamma(\alpha_j)}{\Gamma(\sum_{j=0}^{d-1} \alpha_j)},
\end{equation}
where $\Gamma(\cdot)$ denotes the gamma function. In our case
$\alpha=(1,\dots,1)$, hence
\begin{equation}
f_{\text{Dir}(1,\dots,1)}(p) =  \frac{1}{B(\mathbf{1})} = \frac{\Gamma(\sum_{j=0}^{d-1} 1)}{\prod_{j=0}^{d-1} \Gamma(1)} =\Gamma(d)=(d-1)! \end{equation}
Since for a Haar-distributed unitary matrix $U$ the vector
$W=(W_0,\dots,W_{d-1})$ is uniformly distributed on the simplex
$\Delta_{d-1}$, we can replace the integration over $U$ by integration over the
simplex with the Dirichlet density:
\begin{equation}
    \begin{split}
        & I(\lambda) = \int \prod_{j=0}^{d-1} W_j^{\lambda_j} dU = \int_{\Delta_{d-1}} \prod_{j=0}^{d-1} p_j^{\lambda_j} f_{\text{Dir}(1,\dots,1)}(p) dp.
    \end{split}
\end{equation}
By substituting the values calculated above, we have
\begin{equation} I(\lambda) = (d-1)! \int_{\Delta_{d-1}} \prod_{j=0}^{d-1} p_j^{\lambda_j} dp. \end{equation}

\begin{lem} \label{gamma}
For any parameters $\alpha_0,\dots,\alpha_{d-1}>0$, the following identity holds:
\begin{equation}
\int_{\Delta_{d-1}} \prod_{j=0}^{d-1} p_j^{\alpha_j-1}\,dp= \frac{\prod_{j=0}^{d-1}\Gamma(\alpha_j)}{\Gamma\!\left(\sum_{j=0}^{d-1}\alpha_j\right)}.
\end{equation}

\end{lem}
\begin{proof}
Recall the definition of the Gamma function,
\begin{equation}
\Gamma(s)=\int_0^\infty t^{s-1} e^{-t} dt.
\end{equation}
The product $\prod_{j=0}^{d-1}\Gamma(\alpha_j)$ can therefore be expressed as
\begin{equation}
    \begin{split}
        & \prod_{j=0}^{d-1}\Gamma(\alpha_j)  = \prod_{j=0}^{d-1} \int_0^\infty x_j^{\alpha_j -1} e^{-x_j} dx_j = \int_{[0,\infty)^d}\!
   \Big(\prod_{j=0}^{d-1} x_j^{\alpha_j-1}\Big)
   e^{-\sum_{j=0}^{d-1} x_j}\, dx_0\cdots dx_{d-1}.
    \end{split}
\end{equation}
We now perform the change of variables
\begin{equation}
t = \sum_{j=0}^{d-1} x_j, \qquad p_j = \frac{x_j}{t}, \qquad p\in\Delta_{d-1}.
\end{equation}
The inverse transformation is $x_j = t p_j$, and the Jacobian of this change is
$|\det J| = t^{d-1}$. Hence $dx_0\cdots dx_{d-1} = t^{d-1}dtdp$. Substituting
into the integral yields
\begin{equation}
    \begin{split}
        & \prod_{j=0}^{d-1}\Gamma(\alpha_j) =\int_{t=0}^\infty \int_{p\in\Delta_{d-1}}
         \Big(\prod_{j=0}^{d-1} (tp_j)^{\alpha_j-1}\Big)
         e^{-t} t^{d-1} dp dt= \Big( \int_0^\infty t^{\sum_j \alpha_j-1} e^{-t} dt \Big) \Big(\int_{\Delta_{d-1}} \prod_{j=0}^{d-1} p_j^{\alpha_j-1} dp\Big) \\
        & = \Gamma\left(\sum_{j=0}^{d-1}\alpha_j\right) \int_{\Delta_{d-1}} \prod_{j=0}^{d-1} p_j^{\alpha_j-1} dp.\\
    \end{split}
\end{equation}
Hence
\begin{equation} \int_{\Delta_{d-1}} \prod_{j=0}^{d-1} p_j^{\alpha_j-1} dp = \frac{\prod_{j=0}^{d-1}\Gamma(\alpha_j)}{\Gamma\left(\sum_{j=0}^{d-1}\alpha_j\right)}. \end{equation}
\end{proof}
Using Lemma~\ref{gamma} and $\Gamma(d)=(d-1)!$ we obtain
\begin{equation}
    \begin{split}
        & I(\lambda) = (d-1)! \int_{\Delta_{d-1}} \prod_{j=0}^{d-1} p_j^{\lambda_j} dp =  (d-1)! \frac{\prod_{j=0}^{d-1}\lambda_j!}{\left(\sum_{j=0}^{d-1}(\lambda_j+1)-1\right)!} = (d-1)! \frac{\prod_{j=0}^{d-1} \lambda_j!}{\left(N+d-1\right)!} \\
    \end{split}
\end{equation}
Finally,
\begin{equation}
    \begin{split}
        & \mathbb{P}(\Lambda =\lambda) = m_{\lambda} I(\lambda)  = \frac{N!}{\lambda_0!\lambda_1!\dots \lambda_{d-1}!} (d-1)! \frac{\prod_{j=0}^{d-1} \lambda_j!}{\left(N+d-1\right)!} = \frac{(d-1)!N!}{(N+d-1)!}.
    \end{split}
\end{equation}
\subsection{Proof of Lemma \ref{wspolny}} \label{Pn1Pn2} Assume the measurement
operators act on separate blocks, i.e. $E_{U,\mathbf{l}}^{(n_1)}$ acts on the
first block subsystems $n_1$ and $E_{V,\mathbf{m}}^{(n_2)}$ acts on the second
block $n_2$, and the input state factorizes as
$\rho=\rho^{(n_1)}\otimes\rho^{(n_2)}$ with $\rho^{(n_1)}=\rho^{(n_2)}
=\ket{0\ldots0}\bra{0\ldots0}$. Then
\begin{equation}
    \begin{split}
        & \text{tr}\!\left[(E_{U,\mathbf{l}}^{(n_1)} \otimes E_{V,\mathbf{m}}^{(n_2)})(\rho^{(n_1)} \otimes \rho^{(n_2)})\right] \\
        & = \text{tr}\!\left(E_{U,\mathbf{l}}^{(n_1)}\rho^{(n_1)}\right) \text{tr}\!\left( E_{V,\mathbf{m}}^{(n_2)}\rho^{(n_2)}\right) \\
        & = \bra{0\dots0}E_{U,\mathbf{l}}^{(n_1)} \ket{0\dots0}\bra{0\dots0} E_{V,\mathbf{m}}^{(n_2)}\ket{0\dots0}\\ 
        & = \int \prod_{s=0}^{n_1-1}|\bra{l_s}U\ket{0}|^2 dU \int \prod_{t=0}^{n_2-1}|\bra{m_t}V\ket{0}|^2 dV 
    \end{split}
\end{equation}  
For $Q_{U,V}^{(n_1,n_2)}$ we introduce the count vectors $\lambda^{(1)}$ and
$\lambda^{(2)}$ for the first and the second block, respectively. Computing
analogously as for $Q_U^{(N)}$, we obtain
\begin{equation}
    \begin{split}
        & \mathbb{P}(\Lambda^{(1)}=\lambda^{(1)}, \Lambda^{(2)}=\lambda^{(2)}|U,V) = m^{(1)}(\lambda^{(1)})\prod_{j=0}^{d-1} W_j^{\lambda^{(1)}_j} \cdot m^{(2)}(\lambda^{(2)})\prod_{j=0}^{d-1} Z_j^{\lambda^{(2)}_j}
    \end{split}
\end{equation}
where 
\begin{equation} \begin{split}
    &W_j = |\bra{j}U\ket{0}|^2, \\
    &Z_j = |\bra{j}V\ket{0}|^2.\\
\end{split} \end{equation}
Using the same reasoning as above
\begin{equation}
    \begin{split}
        & \mathbb{P}(\Lambda^{(1)}=\lambda^{(1)},\Lambda^{(2)}=\lambda^{(2)}) = m^{(1)}(\lambda^{(1)}) \cdot I(\lambda^{(1)}) \cdot m^{(2)}(\lambda^{(2)}) \cdot I(\lambda^{(2)})  = \\
        &\frac{n_1!}{\lambda^{(1)}_0!\dots \lambda^{(1)}_{d-1}!}(d-1)!\frac{\prod_{j=0}^{d-1} (\lambda^{(1)}_j)!}{(n_1+d-1)!} \frac{n_2!}{\lambda^{(2)}_0!\dots \lambda^{(2)}_{d-1}!}(d-1)!\frac{\prod_{j=0}^{d-1} (\lambda^{(2)}_j)!}{(n_2+d-1)!} = \frac{n_1! n_2! (d-1)!(d-1)!}{(n_1+d-1)!(n_2+d-1)!}
    \end{split}
\end{equation}
This shows that the joint Haar-averaged distribution factorizes into two
independent multinomial blocks corresponding to the two subsystems.

\subsection{Counting admissible histogram decompositions} \label{liczbapar} We
first ignore the upper bounds and consider $k_j \ge 0$. In this case, the number
of solutions is given by
\begin{equation}
\#\{k_j \ge 0:\ \sum_j k_j = n_1\} = \binom{n_1 + d - 1}{d - 1}.
\end{equation}
We now subtract the solutions that violate the constraints $k_j \le \lambda_j$
using the inclusion–exclusion principle. For each index $j$, let
\begin{equation}
A_j = \{(k_0,\dots,k_{d-1}): k_j > \lambda_j\}
\end{equation}
be the set of solutions where the $j$-th component exceeds its allowed limit. By
the inclusion–exclusion principle,
\begin{equation}
L_{\lambda}(n_1,n_2)= \sum_{S\subseteq\{0,\dots,d-1\}} (-1)^{|S|} |A_S|,
\end{equation}
where \(A_S=\bigcap_{j\in S} A_j\).

For a fixed nonempty subset $S$, the condition $A_S$ means that for each $j\in
S$ we have $k_j \ge \lambda_j+1$. Introducing new variables
\begin{equation}
k'_j =\begin{cases}k_j - (\lambda_j+1), & j\in S,\\k_j, & j\notin S,\end{cases}
\end{equation}
we obtain
\begin{equation}
n_1= \sum_{j \in S} k'_j + (\lambda_j+1) + \sum_{j \notin S} k'_j\end{equation}
\begin{equation}
\sum_{j=0}^{d-1} k'_j = n_1 - \sum_{j\in S} (\lambda_j+1),
\end{equation}
where all $k'_j \ge 0$. Hence, the number of solutions corresponding to \(A_S\)
is
\begin{equation}
|A_S| = \binom{\,n_1 - \sum_{j\in S}(\lambda_j+1) + d - 1\,}{\,d - 1\,},
\end{equation}
with the convention that $\binom{a}{b}=0$ if $a<b$. For $S=\varnothing$ we
recover $|A_\varnothing|=\binom{n_1+d-1}{d-1}$. Thus, the final expression is
\begin{equation} \label{pairs}
    \begin{split}
        & L_{\lambda}(n_1,n_2) = \sum_{S\subseteq\{0,\dots,d-1\}} (-1)^{|S|}  \binom{n_1 - \sum_{j\in S}(\lambda_j+1) + d - 1}{d - 1},
    \end{split}
\end{equation}
where $\binom{a}{b}=0$ for $a<b$.

\begin{prz}
    Let $d=3$ and $\lambda=(2,1,1)$, so that $N=4$. We compute the number of
    pairs for $n_1=3$ and $n_2=1$.
    
    By applying the inclusion-exclusion principle, we have
    \begin{equation}
    L_{\lambda}(3,1)= \sum_{S\subseteq\{0,1,2\}} (-1)^{|S|}
      \binom{3 - \sum_{j\in S}(\lambda_j+1) + 2}{2},
    \end{equation}
    where we assume that \(\binom{a}{b}=0\) for $a<b$. Let us compute the
    contributions for all subsets $S$:
    \begin{itemize}
        \item $S=\varnothing$: $\sum_{j\in S}(\lambda_j+1)=0$,
        \begin{equation}
        (+1)\binom{3+2}{2} = \binom{5}{2}=10.
        \end{equation}
        \item $S=\{0\}$: $\sum_{j\in S}(\lambda_j+1)=\lambda_0+1=2+1=3$, 
        \begin{equation}
        (-1)\binom{3-3+2}{2} = -\binom{2}{2}=-1.
        \end{equation}
        \item $S=\{1\}$: $\lambda_1+1=1+1=2$, 
       \begin{equation}
        (-1)\binom{3-2+2}{2} = -\binom{3}{2}=-3.
        \end{equation}
        \item $S=\{2\}$: analogously to $S=\{1\}$, 
        \begin{equation} -\binom{3}{2}=-3\end{equation}
        \item $S=\{0,1\}$: $\sum_{j\in S}(\lambda_j+1)=3+2=5$,
        \begin{equation}
        (+1)\binom{3-5+2}{2}=\binom{0}{2}=0.
        \end{equation}
        \item $S=\{0,2\}$: analogously we have $0$.
        \item $S=\{1,2\}$: $\sum=2+2=4$,
        \begin{equation}
        (+1)\binom{3-4+2}{2}=\binom{1}{2}=0.
        \end{equation}
        \item $S=\{0,1,2\}$: $\sum=3+2+2=7$,
        \begin{equation}(-1)\binom{3-7+2}{2}=\binom{-2}{2}=0\end{equation}
    \end{itemize}
    Summing all the contributions, we obtain
    \begin{equation}
    L_{\lambda}(3,1) = 10 - 1 - 3 - 3 + 0 + 0 + 0 - 0 = 3.
    \end{equation}
\end{prz}
\section{Proof of Proposition \ref{fixedN}} \label{proofpropo}
Recall that
\begin{equation}
S_0=\{\lambda:\lambda_j\in\{0,1\},\ \sum_j\lambda_j=N\}
\end{equation}
is the set of collision-free histograms, while
\begin{equation}
S_1=\{\lambda:\text{exactly one}, \lambda_j=2,\text{ all others in }\{0,1\}\}
\end{equation}
is the set of single-collision histograms.
The histogram in $S_0$ is specified by choosing the $N$ occupied bins among the $d$ available bins. Hence
\begin{equation}
|S_0|=\binom{d}{N}.
\end{equation}
Similarly, to construct a histogram in $S_1$, one first chooses the bin carrying the double occupation (in $d$ possible ways), and then chooses the remaining $(N-2)$ singly occupied bins among the remaining $(d-1)$ bins. Therefore
\begin{equation}
|S_1|=d\binom{d-1}{N-2}.
\end{equation}
Consequently,
\begin{equation}
    \frac{|S_1|}{|S_0|}=\frac{N(N-1)}{d}+O(d^{-2}).
\end{equation}
When $\lambda$ has no collision ($\lambda \in S_0$), each of the $N$ hits
occupies a distinct bin. Assigning exactly $n_1$ to the first block amounts to
choosing $n_1$ of these $N$ bins, hence
\begin{equation}
L_{S_0}(n_1,n_2)=\binom{N}{n_1}.
\end{equation}
Using the asymptotic expansion
\begin{equation} 
\frac{(d-1)!}{(d+k-1)!} =d^{-k}\Bigl(1-\frac{k(k-1)}{2d} + O(d^{-2})\Bigr),\qquad k=\text{const},
\end{equation}
we obtain
\begin{equation}
    \begin{split}
        & P^{(N)}(\lambda)=N!\,d^{-N}\Bigl(1-\frac{N(N-1)}{2d}+O(d^{-2})\Bigr),\\
        & P_{S_0}^{(n_1,n_2)}(\lambda)= c_{n_1,n_2} \binom{N}{n_1} =N!\,d^{-N}\Bigl(1-\frac{n_1(n_1-1)+n_2(n_2-1)}{2d}+O(d^{-2})\Bigr).\\
    \end{split}
\end{equation}
where
\begin{equation}
c_{n_1,n_2} :=\frac{n_1! \, n_2! \, (d-1)!^2}{(n_1+d-1)! \, (n_2+d-1)!}.
\end{equation}
Thus
\begin{equation}
\begin{split}
    & \bigl|P^{(N)}(\lambda)-P_{S_0}^{(n_1,n_2)}(\lambda)\bigr| =\left| N! d^{-N} \left( 1-\frac{N(N-1)}{2d} -  1 +\frac{n_1(n_1-1)+n_2(n_2-1)}{2d} \right) +O(d^{-N-2}) \right|\\
    & = \left| -N! d^{-N} \left( \frac{N(N-1)-n_1(n_1-1) - n_2(n_2-1)}{2d}\right) +O(d^{-N-2}) \right| \\
    & = \left| -N! d^{-N} \frac{N^2 - n_1^2 -n_2^2}{2d} + O(d^{-N-2}) \right| = \left| -N! d^{-N} \frac{2n_1n_2}{2d} + O(d^{-N-2}) \right| \\
    &= N!\,d^{-N}\frac{n_1n_2}{d}+O(d^{-N-2}),
\end{split}
\end{equation}
Summing over all $\lambda\in S_0$ gives
\begin{equation}
\begin{split}
    & \mathrm{TVD}_{S_0}=\frac{1}{2}|S_0|\,N!d^{-N}\frac{n_1n_2}{d} = \frac{1}{2} \binom{d}{N} N!d^{-N}\frac{n_1n_2}{d} = \frac{n_1n_2}{2d}\frac{d!N!}{d^N N!(d-N)!}\\
    &=\frac{n_1n_2}{2d} \frac{d(d-1)\cdots (d-N+1)}{d^N} = \frac{n_1n_2}{2d} \prod_{i=0}^{N-1} \Big(1 - \frac{i}{d} \Big) =\frac{n_1n_2}{2d}(1+O(d^{-1}))= \frac{n_1n_2}{2d}+O(d^{-2}).
\end{split}
\end{equation}
The expansion holds uniformly for every fixed $N$. We now turn to histograms
with exactly one collision. For a histogram with one bin equal to $2$ ($\lambda
\in S_1$),
\begin{equation}
L_{S_1}(n_1,n_2)=\binom{N-2}{n_1}+\binom{N-2}{n_1-1}+\binom{N-2}{n_1-2}= \binom{N}{n_1}-\binom{N-2}{n_1-1}.
\end{equation}
Thus
\begin{equation}
\begin{split}
&P_{S_1}^{(n_1,n_2)}(\lambda) = c_{n_1,n_2}L_{S_1}(n_1,n_2)\\
&= n_1!n_2!\,d^{-N}\left(\binom{N}{n_1}-\binom{N-2}{n_1-1}\right)\Bigl(1-\frac{n_1(n_1-1)+n_2(n_2-1)}{2d}+O(d^{-2})\Bigr)\\
&= N!\,d^{-N}\Bigl(1-\frac{n_1(n_1-1)+n_2(n_2-1)}{2d}\Bigr)- n_1!n_2!\,d^{-N}\binom{N-2}{n_1-1}+O(d^{-N-2}).
\end{split}
\end{equation}
and
\begin{equation}
\begin{split}
    & P^{(N)}(\lambda)-P_{S_1}^{(n_1,n_2)}(\lambda)= N!d^{-N}\Bigl(\frac{-N(N-1)+n_1(n_1- 1)+n_2(n_2-1)}{2d}\Bigr)+\\
    &+ n_1!n_2!\,d^{-N}\binom{N-2}{n_1-1} + O(d^{-N-2}) = -N!d^{-N}\frac{n_1n_2}{d} + n_1!n_2!\,d^{-N}\binom{N-2}{n_1-1} + O(d^{-N-2}).
    \end{split}
\end{equation}
We now evaluate the contribution of single-collision histograms $\lambda \in
S_1$. Using $|S_1|=d\binom{d-1}{N-2}$ we get
\begin{equation}
\begin{split}
&\mathrm{TVD}_{S_1}
=\frac{1}{2}|S_1|\Bigg|N!d^{-N}\frac{n_1n_2}{d}- n_1!n_2!\,d^{-N}\binom{N-2}{n_1-1}+ O(d^{-N-2})\Bigg|\\
&= \Bigg|\frac{1}{2}\,d\binom{d-1}{N-2}N!d^{-N}\frac{n_1n_2}{d}- \frac{1}{2}\,d\binom{d-1}{N-2}\,n_1!n_2!\,d^{-N}\binom{N-2}{n_1-1}\Bigg|+ O(d^{-2}) =: |T_A-T_B| + O(d^{-2})
\end{split}
\end{equation}
The first term becomes
\begin{equation}
    \begin{split}
        & T_A = \frac{1}{2}\,d\binom{d-1}{N-2}\,N!d^{-N}\frac{n_1n_2}{d} = \frac{n_1n_2}{2d}\frac{d(d-1)\cdots(d-N+2)}{d^{N-1}}\frac{N(N-1)}{d} \\
        & =  \frac{n_1n_2}{2d} \frac{N(N-1)}{d} \bigl(1+O(d^{-1})\bigr) =   \frac{n_1n_2N(N-1)}{2d^2}(1+O(d^{-1})) = O(d^{-2})
    \end{split}
\end{equation}
The second term
\begin{equation}
    \begin{split}
        & T_B = \frac{1}{2}\,d\binom{d-1}{N-2}\,n_1!n_2!\,d^{-N}\binom{N-2}{n_1-1} =\frac{1}{2}\frac{1}{(N-2)!}\frac{1}{d}\bigl(1+O(d^{-1})\bigr)\frac{(N-2)!n_1!n_2!}{(n_1-1)!(n_2-1)!}\\
        & =\frac{n_1n_2}{2d}(1+O(d^{-1})) = \frac{n_1n_2}{2d} +O(d^{-2})
    \end{split}
\end{equation}
Collecting terms we have
\begin{equation}
    \begin{split}
        & \mathrm{TVD}_{S_1}= \left|O\left(d^{-2}\right) - \Big(\frac{n_1n_2}{2d}+O\left(d^{-2}\right)\Big)\right| + O\left(d^{-2}\right) = \frac{n_1n_2}{2d} + O\left(d^{-2}\right) 
    \end{split}
\end{equation}
Summing the two leading contributions yields
\begin{equation}
\begin{split}
    &\mathrm{TVD}\left(P^{(N)},P^{(n_1,n_2)}\right)= \mathrm{TVD}_{S_0} + \mathrm{TVD}_{S_1} + O(d^{-2}) =\frac{n_1n_2}{2d}+O(d^{-2}) + \frac{n_1n_2}{2d} + O(d^{-2})\\ &= \frac{n_1n_2}{d} + O(d^{-2})
\end{split}
\end{equation}

\section{Proof of Lemma \ref{Postacpsi}} \label{erhart}
We use the inclusion--exclusion formula
\begin{equation}
L_{\lambda}(n_1,n_2)=\sum_{S\subseteq\{0,\dots,d-1\}}(-1)^{|S|}
\binom{n_1-\sum_{j\in S}(\lambda_j+1)+d-1}{d-1},
\end{equation}
with the convention that the binomial coefficient is zero if the upper argument
is smaller than $d-1$. For fixed $d\ge2$ and all integers $a=O(n_1)$, uniformly
including the boundary region $a=0$, we have
\begin{equation}
\binom{a+d-1}{d-1}\mathbf 1_{\{a\ge0\}}
=\frac{a_+^{\,d-1}}{(d-1)!}+O(n_1^{d-2}).
\end{equation}
Indeed, for $a\ge0$ this follows from expanding
$(a+1)\cdots(a+d-1)/(d-1)!$ as a polynomial in $a$, while for $a<0$ both leading
terms vanish. Applying this estimate to
\begin{equation}
a_S:=n_1-\sum_{j\in S}(\lambda_j+1)
\end{equation}
and summing over the finitely many subsets $S$ gives the uniform expansion
\begin{equation}
L_{\lambda}(n_1,n_2)=\frac{1}{(d-1)!}
\sum_{S\subseteq\{0,\dots,d-1\}}(-1)^{|S|}
\left(n_1-\sum_{j\in S}(\lambda_j+1)\right)_+^{d-1}
O(n_1^{d-2}).
\end{equation}
Let $\ell_j=\lambda_j/N$ and $\alpha_N=n_1/N$. Since
$n_1=\lfloor\alpha N\rfloor$, we have $\alpha_N\to\alpha$ and
\begin{equation}
\frac{\sum_{j\in S}(\lambda_j+1)}{n_1}
=\sum_{j\in S}\frac{\ell_j}{\alpha_N}+\frac{|S|}{n_1}.
\end{equation}
The function $x\mapsto (1-x)_+^{d-1}$ is Lipschitz on every bounded interval, and
the quantities above remain uniformly bounded for $\ell\in\Delta_{d-1}$. Hence,
uniformly in $\lambda$ and in $S$,
\begin{equation}
\left(1-\frac{\sum_{j\in S}(\lambda_j+1)}{n_1}\right)_+^{d-1}
=\left(1-\sum_{j\in S}\frac{\ell_j}{\alpha}\right)_+^{d-1}
O(n_1^{-1}).
\end{equation}
After multiplication by $n_1^{d-1}$ and summation over $S$, this error is
$O(n_1^{d-2})$. Finally, replacing $\ell_j/\alpha$ by
$\min\{1,\ell_j/\alpha\}$ does not change any term: if some $\ell_j/\alpha\ge1$
appears in a subset $S$, then the corresponding positive part is already zero;
otherwise the minimum leaves the term unchanged. Therefore
\begin{equation}
L_{\lambda}(n_1,n_2)=\frac{n_1^{d-1}}{(d-1)!}\,
\varphi_\alpha(\ell)+O(n_1^{d-2}),
\end{equation}
uniformly over all histograms $\lambda$, where
\begin{equation}
\varphi_\alpha(\ell)=
\sum_{S\subseteq\{0,\dots,d-1\}}(-1)^{|S|}
\left(1-\sum_{j\in S}\min\left\{1,\frac{\ell_j}{\alpha}\right\}\right)_+^{d-1}.
\end{equation}

\section{Proof of Proposition \ref{JSsparse}} \label{sparseproof}
Observe that
\begin{equation}
    \binom{k+d-1}{d-1}= \frac{d(d+1)\cdots(d+k-1)}{k!} = \frac{d^k}{k!} \prod_{j=0}^{k-1} \left(1 + \frac{j}{d}\right).
\end{equation}
Taking logarithms, we obtain
\begin{equation}
    \log \prod_{j=0}^{k-1} \left(1 + \frac{j}{d}\right) = \sum_{j=0}^{k-1} \log\left(1 + \frac{j}{d}\right).
\end{equation}
Using $\log(1+x) = x - \frac{x^2}{2} + O(x^3)$, valid for small $x$, we get
\begin{equation}
\begin{split}
    &\sum_{j=0}^{k-1} \log\left(1 + \frac{j}{d}\right)= \sum_{j=0}^{k-1} \left(\frac{j}{d} - \frac{j^2}{2d^2} + O\left(\frac{j^3}{d^3}\right)\right)=\frac{1}{d}\sum_{j=0}^{k-1}j-\frac{1}{2d^2}\sum_{j=0}^{k-1}j^2+ O\left(\frac{1}{d^3} \sum_{j=0}^{k-1} j^3\right)\\
    &=\frac{k(k-1)}{2d}-\frac{k(k-1)(2k-1)}{12d^2}+O\left(\frac{k^4}{d^3}\right)=\frac{k(k-1)}{2d}+O\left(\frac{k^3}{d^2}\right)+O\left(\frac{k^4}{d^3}\right)=\frac{k(k-1)}{2d}+O\left(\frac{k^3}{d^2}\right).
     \end{split}
\end{equation}
Therefore
  \begin{equation}
      \begin{split}
          &\log \prod_{j=0}^{k-1} \left(1 + \frac{j}{d}\right)= \frac{k(k-1)}{2d}+ O\left(\frac{k^3}{d^2}\right).
      \end{split}
  \end{equation}
Exponentiating, and using $ e^x = 1 + x + O(x^2)$ for small $x$, we obtain
\begin{equation}
    \begin{split}
        &\prod_{j=0}^{k-1} \left(1 + \frac{j}{d}\right)= \exp\left(\frac{k(k-1)}{2d}+O\left(\frac{k^3}{d^2}\right)\right)=\exp\left(\frac{k(k-1)}{2d}\right)\exp\left(O\left(\frac{k^3}{d^2}\right)\right)\\
        &=\left( 1 + \frac{k(k-1)}{2d}+ O\left(\frac{k^4}{d^2}\right)\right)\left(1+O\left(\frac{k^3}{d^2}\right)\right)=1 + \frac{k(k-1)}{2d}+ O\left(\frac{k^4}{d^2}\right).
   \end{split}
\end{equation}
Here we used that $\left(\frac{k(k-1)}{2d}\right)^2 = O\left(\frac{k^4}{d^2}\right)$ dominates the error term $O\left(\frac{k^3}{d^2}\right)$.
Hence
\begin{equation}
    \binom{k+d-1}{d-1}= \frac{d^k}{k!} \left( 1 + \frac{k(k-1)}{2d}+ O\left(\frac{k^4}{d^2}\right)\right).
\end{equation}
Applying this asymptotic expansion together with
\begin{equation}
    \frac{1+a}{1+b} =(1+a)\frac{1}{1+b}=(1+a)(1-b+O(b^2))= 1 + (a-b) + O(a^2 + b^2),
\end{equation}
for sufficiently small $a,b$, we obtain the expansion for $\overline M$
\begin{equation}
    \begin{split}
        & \overline{M}=\frac{M_1M_2}{M}=\frac{N!}{n_1!n_2!}\frac{1+\frac{n_1(n_1-1)+n_2(n_2-1)}{2d}}{1+\frac{N(N-1)}{2d}}+O\left(\frac{N^4}{d^2}\right)\\
        & =\binom{N}{n_1}\left(1+\frac{1}{d}\left(\frac{n_1(n_1-1)+n_2(n_2-1)}{2}-\frac{N(N-1)}{2}\right)+O(d^{-2})\right)+O\left(\frac{N^4}{d^2}\right)\\
        & = \binom{N}{n_1}\left(1-\frac{n_1n_2}{d}+O\left(\frac{N^4}{d^2}\right)\right).
    \end{split}
\end{equation}
The same stratification by collision structure as in Proposition~\ref{fixedN}
remains valid uniformly when $N=o(\sqrt{d})$. For collision-free histograms
$S_0$, each outcome appears at most once, and the decomposition of the total
histogram into two blocks reduces to choosing which $n_1$ outcomes belong to the
first block. Then
\begin{equation}
L_{S_0}(n_1,n_2)=\binom{N}{n_1},
\end{equation}
Hence, using $\frac{1}{1-x} = 1 + x + O(x^2)$ with $x = \frac{n_1 n_2}{d}$,
\begin{equation} \label{zerocollisions}
    \frac{L_{S_0}(n_1,n_2)}{\overline{M}}=\frac{1}{1-\frac{n_1n_2}{d}+O\left(\frac{N^4}{d^2}\right)}=1+\frac{n_1n_2}{d}+O\left(\frac{N^4}{d^2}\right)
\end{equation}
For histograms in $S_1$ with a single double occupation, the combinatorics is
slightly modified, since the doubleton may be assigned entirely to one block or
split across the two blocks.
\begin{equation}
L_{S_1}(n_1,n_2)=\binom{N-2}{n_1}+\binom{N-2}{n_1-1}+\binom{N-2}{n_1-2}
=\binom{N}{n_1}-\binom{N-2}{n_1-1}.
\end{equation}
The drop is
\begin{equation}
\Delta_{S_1}=\binom{N}{n_1}-L_{S_1}(n_1,n_2) = \binom{N-2}{n_1-1}=\binom{N}{n_1}\frac{n_1 n_2}{N(N-1)},
\end{equation}
Hence
\begin{equation}
\frac{\Delta_{S_1}}{\overline{M}}= \frac{\binom{N}{n_1}\frac{n_1 n_2}{N(N-1)}}{\binom{N}{n_1}\left(1-\frac{n_1n_2}{d}\right)}+O\left(\frac{N^4}{d^2}\right)=\frac{n_1 n_2}{N(N-1)}\left(1+\frac{n_1n_2}{d}\right)+O\left(\frac{N^4}{d^2}\right)\end{equation}
Stratum sizes:
\begin{equation}
|S_0|=\binom{d}{N},\qquad
|S_1|=d\binom{d-1}{N-2}=|S_0|\frac{N(N-1)}{d-N+1}.
\end{equation}
Moreover, the contribution of histograms with at least two collisions is
negligible in the sparse regime. Indeed, such configurations require either two
distinct double occupations or one occupation number at least three. Their
respective cardinalities satisfy
\begin{equation}
|S_{2}|=\binom{d}{2}\binom{d-2}{N-4},\qquad|S_3|=d\binom{d-1}{N-3},
\end{equation}
and therefore
\begin{equation}
\mathbb P(S_{\ge2\mathrm{coll.}})=O\!\left(\frac{N^4}{d^2}\right),
\end{equation}
uniformly in the regime \(N=o(\sqrt d)\). Putting this into the
uniform-expectation identity gives
\begin{equation}
    \begin{split}
        & \mathrm{TVD}\left(P^{(N)},P^{(n_1,n_2)}\right)= \frac12 \mathbb{E}_{\lambda}\Bigg[ \Bigg|1-\frac{L_{\lambda}(n_1,n_2)}{\overline{M}}\Bigg|\Bigg]= \frac12 \Bigg(\mathbb{P}(S_0) \mathbb{E}\Bigg[ \Bigg|1-\frac{L_{S_0}(n_1,n_2)}{\overline{M}}\Bigg| \Big| S_0\Bigg]+\\ 
        & + \mathbb{P}(S_1)\mathbb{E}\Bigg[ \Bigg|1-\frac{L_{S_1}(n_1,n_2)}{\overline{M}}\Bigg| \Big| S_1\Bigg]
         + \mathbb{P}(S_{\ge2\mathrm{coll.}})\mathbb{E}\Bigg[ \Bigg|1-\frac{L_{\ge2\mathrm{coll.}}(n_1,n_2)}{\overline{M}}\Bigg| \Big|S_{\ge2\mathrm{coll.}}\Bigg]\Bigg)\\
        & =\frac12\mathbb{P}(S_0)\mathbb{E}\Bigg[ \Bigg|1-\frac{L_{S_0}(n_1,n_2)}{\overline{M}}\Bigg|\Big| S_0\Bigg] + \frac12 \mathbb{P}(S_1)\mathbb{E}\Bigg[ \Bigg|1-\frac{L_{S_1}(n_1,n_2)}{\overline{M}}\Bigg|\Big| S_1\Bigg] + O\left( \frac{N^4}{d^2} \right)\\
    \end{split}
\end{equation}
Since the integrand is constant on the strata $S_0$ and $S_1$, the corresponding
conditional expectations reduce to deterministic values. 
\begin{equation}
    \mathrm{TVD}\left(P^{(N)},P^{(n_1,n_2)}\right)=\frac12\mathbb{P}(S_0)\left|1-\frac{L_{S_0}(n_1,n_2)}{\overline{M}}\right|+ \frac12 \mathbb{P}(S_1) \left|1-\frac{L_{S_1}(n_1,n_2)}{\overline{M}}\right| + O\left( \frac{N^4}{d^2} \right)
\end{equation}
Since histograms are uniformly distributed, the probability of observing zero or
one collision is given by
\begin{equation}
    \begin{split}
        & \mathbb{P}(S_0)=\frac{|S_0|}{M}=\frac{\binom{d}{N}}{M}=\frac{d^N}{N!}\left(1-\frac{N(N-1)}{2d}+O\left( \frac{N^4}{d^2}\right)\right)\cdot\frac{N!}{d^N}\left(1+\frac{N(N-1)}{2d}+O\left( \frac{N^4}{d^2}\right)\right)\\
        &=\frac{1-\frac{N(N-1)}{2d}}{1+\frac{N(N-1)}{2d}}+O\left( \frac{N^4}{d^2}\right) =1-\frac{N(N-1)}{d}+O\left( \frac{N^4}{d^2}\right)
    \end{split}
\end{equation}
and
\begin{equation}
    \begin{split}
        &\mathbb P(S_1)=\frac{|S_1|}{M}= \frac{|S_0|}{M}\frac{N(N-1)}{d-N+1}=\left(1-\frac{N(N-1)}{d}+O\left( \frac{N^4}{d^2} \right)\right)\frac{N(N-1)}{d-N+1}\\
        &=\left(1-\frac{N(N-1)}{d}+O\left( \frac{N^4}{d^2} \right)\right) \frac{N(N-1)}{d}\left(\frac{1}{1-\frac{N-1}{d}}\right)\\
        &= \frac{N(N-1)}{d}\left(1-\frac{N(N-1)}{d}+O\left( \frac{N^4}{d^2} \right)\right) \left(1+\frac{N-1}{d}+O\left(\frac{N^2}{d^2}\right)\right) \\
        &=\frac{N(N-1)}{d}\left(1-\frac{N(N-1)}{d}+O\left( \frac{N^4}{d^2} \right)\right)\left(1+O\left(\frac{N}{d}\right)\right)\\
        &=\frac{N(N-1)}{d}+O\left( \frac{N^3}{d^2} \right)+O\left( \frac{N^4}{d^2} \right)=\frac{N(N-1)}{d} +O\left( \frac{N^4}{d^2} \right)
    \end{split}
\end{equation}
where $M=\binom{N+d-1}{d-1}$ is the total number of histograms. Then
\begin{equation}
    \begin{split}
        &\mathrm{TVD}\left(P^{(N)},P^{(n_1,n_2)}\right)=\frac{1}{2}\frac{|S_0|}{M}\frac{n_1n_2}{d}+ \frac12 \frac{|S_1|}{M} \frac{\Delta_{S_1}}{\overline{M}} + O\left( \frac{N^4}{d^2} \right) \\
        & = \frac12\left(1-\frac{N(N-1)}{d}\right) \frac{n_1n_2}{d}+\frac12\frac{N(N-1)}{d}\frac{n_1 n_2}{N(N-1)}\left(1+\frac{n_1n_2}{d}\right) + O\left( \frac{N^4}{d^2} \right)\\
        &= \frac{n_1n_2}{2d}-\frac{n_1n_2N(N-1)}{2d^2}+\frac{n_1n_2}{2d}+\frac{n_1^2n_2^2}{2d^2}+O\!\left(\frac{N^4}{d^2}\right)\\
        & =\frac{n_1n_2}{2d}+\frac{n_1n_2}{2d}+O\!\left(\frac{N^4}{d^2}\right)= \frac{n_1n_2}{d}+O\!\left(\frac{N^4}{d^2}\right).
    \end{split}
\end{equation}
\section{Critical scaling regime}
\subsection{Proof of Proposition~\ref{Poisson}} \label{lkolizji}
The robust route is to work with the number of doubletons
\begin{equation}
    D_2(\lambda):=\#\{i:\lambda_i=2\},
\end{equation}
and to separate higher occupancies through the event
\begin{equation}
A_d:=\left\{\max_i\lambda_i\le 2\right\}.
\end{equation}
On $A_d$, we have $C(\lambda)=D_2(\lambda)$.

\paragraph{Step 1: higher occupancies are rare.}
Under the uniform composition law, $\lambda$ has the Dirichlet-multinomial
representation
\begin{equation}
W\sim \mathrm{Dirichlet}(1,\dots,1),\qquad (Y_1,\dots,Y_N)\mid W \text{ i.i.d. with law }W.
\end{equation}
Hence
\begin{equation}
\mathbb P(Y_1=Y_2=Y_3)=\mathbb E\!\left[\sum_i W_i^3\right]=\frac{6}{(d+1)(d+2)}.
\end{equation}
Therefore
\begin{equation}
\mathbb P(A_d^c)\le \mathbb E\!\left[\sum_i\binom{\lambda_i}{3}\right]=\binom{N}{3}\frac{6}{(d+1)(d+2)}.
\end{equation}
For $N/\sqrt d\to c$, this gives
\begin{equation}
\mathbb P(A_d^c)=O\left(d^{-\frac12}\right)\to 0.
\end{equation}

\paragraph{Step 2: Poisson limit for $D_2(\lambda)$.}
For fixed $r\ge 1$, the falling factorial moment is
\begin{equation}
\mathbb E[(D_2)_r]=(d)_r\,\frac{\binom{N-2r+d-r-1}{d-r-1}}{\binom{N+d-1}{d-1}},
\end{equation}
where $(d)_j = \prod_{l=0}^{j-1}(d-l)$ is the Pochhammer symbol. This identity is
exact: choose $r$ ordered bins to carry occupancy $2$, then distribute the
remaining $N-2r$ balls among the other $d-r$ bins. With $N/\sqrt d\to c$ and fixed
$r$,
\begin{equation}
\mathbb E[(D_2)_r]\longrightarrow c^{2r}.
\end{equation}
Thus $D_2(\lambda)\Rightarrow \mathrm{Poisson}(c^2)$ by the method of factorial
moments.

\paragraph{Step 3: transfer from $D_2$ to $C$.}
Since $C(\lambda)=D_2(\lambda)$ on $A_d$ and $\mathbb P(A_d^c)\to 0$,
\begin{equation}
\mathbb P\!\big(C(\lambda)\neq D_2(\lambda)\big)\to 0.
\end{equation}
Combining this with $D_2(\lambda)\Rightarrow \mathrm{Poisson}(c^2)$ yields
\begin{equation}
C(\lambda)\Rightarrow \mathrm{Poisson}(c^2),
\end{equation}
which proves Proposition~\ref{Poisson}.

\subsection{Proof of Proposition~\ref{grLn}} \label{L/binom} Assume $\lambda$
has exactly $k$ bins of multiplicity $2$ and no higher multiplicities. Then
\begin{equation}
L_{\lambda}(n_1,n_2)=[x^{n_1}](1+x)^{N-2k}(1+x+x^2)^k.
\end{equation}
Using $1+x+x^2=(1+x)^2-x$,
\begin{equation}
\begin{aligned}
L_{\lambda}(n_1,n_2)
&=[x^{n_1}](1+x)^{N-2k}\sum_{j=0}^k(-1)^j\binom{k}{j}x^j(1+x)^{2k-2j}=\sum_{j=0}^k(-1)^j\binom{k}{j}\binom{N-2j}{n_1-j}.
\end{aligned}
\end{equation}
Hence
\begin{equation}
\begin{aligned}
\frac{L_{\lambda}(n_1,n_2)}{\binom{N}{n_1}}& =\sum_{j=0}^k(-1)^j\binom{k}{j}\frac{\binom{N-2j}{n_1-j}}{\binom{N}{n_1}}=\sum_{j=0}^k(-1)^j\binom{k}{j}\frac{n_1\cdots(n_1-j+1)n_2\cdots(n_2-j+1)}{N\cdots(N-2j+1)}\\&=\sum_{j=0}^k(-1)^j\binom{k}{j}\frac{(n_1)_j(n_2)_j}{(N)_{2j}},
\end{aligned}
\end{equation}
where $(m)_j = \prod_{l=0}^{j-1}(m-l)$. For each fixed $j$, if $n_1/N\to\alpha$ and $n_2/N\to1-\alpha$, then
\begin{equation}
\frac{(n_1)_j(n_2)_j}{(N)_{2j}}\longrightarrow \alpha^j(1-\alpha)^j.
\end{equation}
Therefore, for fixed $k$,
\begin{equation}
\frac{L_{\lambda}(n_1,n_2)}{\binom{N}{n_1}}\longrightarrow\sum_{j=0}^k(-1)^j\binom{k}{j}\big(\alpha(1-\alpha)\big)^j=\big(1-\alpha(1-\alpha)\big)^k.
\end{equation}
This proves Proposition~\ref{grLn}.

\subsection{Proof of normalization Lemma~\ref{normalizationlem}} \label{normalization}
Write
\begin{equation}
\overline{M}=\frac{M_1M_2}{M},
\qquad
M=\binom{N+d-1}{N},
\quad
M_1=\binom{n_1+d-1}{n_1}, 
\quad 
M_2=\binom{n_2+d-1}{n_2}.
\end{equation}
For $k=O(\sqrt d)$,
\begin{equation}
\binom{d+k-1}{k}=\frac{d^k}{k!}\exp\!\left(\frac{k(k-1)}{2d}+O\!\left(\frac{k^3}{d^2}\right)\right).
\end{equation}
Applying this to $k=N,n_1,n_2$ gives
\begin{equation}
M=\frac{d^N}{N!}\exp\!\left(\frac{N(N-1)}{2d}+O\!\left(\frac{N^3}{d^2}\right)\right),
\end{equation}
\begin{equation}
M_i=\frac{d^{n_i}}{n_i!}\exp\!\left(\frac{n_i(n_i-1)}{2d}+O\!\left(\frac{N^3}{d^2}\right)\right),\quad i=1,2.
\end{equation}
Hence
\begin{equation}
\overline{M}=\binom{N}{n_1}\exp\!\left(-\frac{n_1n_2}{d}+O\!\left(\frac{N^3}{d^2}\right)\right),
\end{equation}
so
\begin{equation}
\frac{\binom{N}{n_1}}{\overline{M}}=\exp\!\left(\frac{n_1n_2}{d}+O\!\left(\frac{N^3}{d^2}\right)\right).
\end{equation}
In the critical scaling $N/\sqrt d\to c$ with $n_1/N\to\alpha$,
\begin{equation}
\frac{\binom{N}{n_1}}{\overline{M}} \longrightarrow e^{\alpha(1-\alpha)c^2}.
\end{equation}
This proves Lemma~\ref{normalizationlem}.

\subsection{Proof of Theorem~\ref{twCritScal}} \label{ProofCritical}
Recall
\begin{equation}
\mathrm{TVD}\left(P^{(N)},P^{(n_1,n_2)}\right) =\frac12\,\mathbb E_{\lambda\sim P^{(N)}}\!\left[\left|1-\frac{L_{\lambda}(n_1,n_2)}{\overline{M}}\right|\right].
\end{equation}
On the event $A_d$, a histogram with $D_2(\lambda)=k$ has only singletons and
doubletons. For each fixed $k$, Proposition~\ref{grLn} gives
\begin{equation}
\frac{L_{\lambda}(n_1,n_2)}{\binom{N}{n_1}} \longrightarrow (1-\alpha(1-\alpha))^k.
\end{equation}
Fix $\varepsilon>0$. Choose $R$ large enough so that
\begin{equation}
    \mathbb P(D_2(\lambda)>R)<\varepsilon
\end{equation}
for all sufficiently large $d$, using the convergence $D_2(\lambda)\Rightarrow
\mathrm{Poisson}(c^2)$. Since $R$ is fixed, Proposition~\ref{grLn} applies to
each $k\le R$, and taking the maximum over finitely many values of $k$ yields
uniform convergence on $\{0,\dots,R\}$. Hence
\begin{equation}
  \sup_{\lambda:\, D_2(\lambda)\le R}\left|\frac{L_\lambda(n_1,n_2)}{\binom{N}{n_1}}-(1-\alpha(1-\alpha))^{D_2(\lambda)}\right|\longrightarrow 0,  
\end{equation}
where the supremum is taken over histograms satisfying $D_2(\lambda)=k$.
Therefore, uniformly on the event $A_d\cap\{D_2(\lambda)\le R\}$,
\begin{equation}
    \left|\frac{L_\lambda(n_1,n_2)}{\binom{N}{n_1}}-(1-\alpha(1-\alpha))^{D_2(\lambda)}\right|  \longrightarrow 0 .
\end{equation}
Since $\mathbb P(A_d^c)\to0$ and $\mathbb P(D_2(\lambda)>R)<\varepsilon$, it follows that
\begin{equation}
    \frac{L_\lambda(n_1,n_2)}{\binom{N}{n_1}}-(1-\alpha(1-\alpha))^{D_2(\lambda)}\xrightarrow[]{p} 0.
\end{equation}
Combining this with Lemma~\ref{normalizationlem},
\begin{equation}
\frac{L_\lambda(n_1,n_2)}{\overline M}-e^{\alpha(1-\alpha)c^2}(1-\alpha(1-\alpha))^{D_2(\lambda)}\xrightarrow[]{p} 0.
\end{equation}
Since $D_2(\lambda)\Rightarrow K$ with $K\sim\mathrm{Poisson}(c^2)$,
\begin{equation}
e^{\alpha(1-\alpha)c^2}(1-\alpha(1-\alpha))^{D_2(\lambda)}\xrightarrow[]{d} e^{\alpha(1-\alpha)c^2}(1-\alpha(1-\alpha))^K.
\end{equation}
Since the difference between the two quantities converges to zero in
probability, we have
\begin{equation}
\frac{L_\lambda(n_1,n_2)}{\overline M}\xrightarrow[]{d} e^{\alpha(1-\alpha)c^2}(1-\alpha(1-\alpha))^K.
\end{equation}
By Lemma~\ref{criticaldomination}, the corresponding absolute deviations are
uniformly integrable. Therefore convergence in distribution upgrades to
convergence of the expectation:
\begin{equation}
\frac12\,\mathbb E\!\left[\left|1-\frac{L_{\lambda}(n_1,n_2)}{\overline{M}}\right|\right]\longrightarrow\frac12\,\mathbb E\!\left[\left|1-e^{\alpha(1-\alpha) c^2}(1-\alpha(1-\alpha))^K\right|\right].
\end{equation}
This is exactly Theorem~\ref{twCritScal}.

\section{Block-resolved asymptotic}
\subsection{Proof of Proposition~\ref{SparseResolved}} \label{ProofSR}

We follow the collision decomposition used in Proposition~\ref{fixedN} for the
aggregate statistic, now applied to the pair histogram $(a,b)$. Recall that the
likelihood ratio is
\begin{equation}
\frac{P_{\mathrm{sh}}(a,b)}{P_{\mathrm{ind}}(a,b)}=\frac{M_{1}M_{2}}{M\binom{N}{n_1}}\prod_{j=0}^{d-1}\binom{a_j+b_j}{a_j}.
\end{equation}
Hence, by \eqref{pairtvdlr},
\begin{equation}
\mathrm{TVD}\left(P_{\mathrm{sh}},P_{\mathrm{ind}}\right)=\frac12\mathbb E_{\mathrm{ind}}\left[\left|\frac{P_{\mathrm{sh}}(a,b)}{P_{\mathrm{ind}}(a,b)}-1\right|\right].
\end{equation}
We decompose the configuration space according to the collision structure.
Define
\begin{align}
S_0
&=\left\{(a,b):a_j+b_j\in\{0,1\}\text{ for all } j\right\},\\[1ex]
S_{\times}
&=
\left\{(a,b):
\begin{array}{l}
\exists!_j \text{ such that } (a_j,b_j)=(1,1),\\
\text{and all remaining coordinates belong to }\\
\{(1,0),(0,1),(0,0)\}
\end{array}
\right\},
\\[1ex]
S_{\mathrm{int}}
&=
\left\{
(a,b):
\begin{array}{l}
\exists!_j \text{ such that }
(a_j,b_j)\in\{(2,0),(0,2)\},\\
\text{and all remaining coordinates belong to }\\
\{(1,0),(0,1),(0,0)\}
\end{array}
\right\}.
\end{align}
and let $S_{\ge2}$ denote the remaining configurations, that is configurations containing either at least two collision events or a non-simple collision pattern (such as occupancies of size at least three). Accordingly, we write
\begin{equation}
\mathrm{TVD}\left(P_{\mathrm{sh}},P_{\mathrm{ind}}\right)=\mathrm{TVD}_{S_0}+\mathrm{TVD}_{S_\times}+\mathrm{TVD}_{S_{\mathrm{int}}}+\mathrm{TVD}_{S_{\ge2}},
\end{equation}
where for any sector \(S\),
\begin{equation}
\mathrm{TVD}_S:=\frac12\sum_{(a,b)\in S}P_{\mathrm{ind}}(a,b)\left|\frac{P_{\mathrm{sh}}(a,b)}{P_{\mathrm{ind}}(a,b)}-1\right|.
\end{equation}
We may therefore decompose the expectation according to the collision sectors as follows:
\begin{equation}
\begin{split}
\mathrm{TVD}\left(P_{\mathrm{sh}},P_{\mathrm{ind}}\right)
& = \frac12\sum_{S\in\{S_0,S_\times,S_{\mathrm{int}},S_{\ge2}\}}
\mathbb P_{\mathrm{ind}}(S)  \mathbb E_{\mathrm{ind}}\left[\left|\frac{P_{\mathrm{sh}}(a,b)}{P_{\mathrm{ind}}(a,b)}-1\right| \,\middle|\,
(a,b)\in S \right].
\end{split}
\end{equation}
We first estimate the contribution of $S_{\ge2}$. Exactly as in
Section~\ref{collision}, every additional collision reduces by one the number of
freely chosen occupied bins. Therefore
\begin{equation}
|S_{\ge2}|=O(d^{N-2}).
\end{equation}
Moreover, uniformly in $(a,b)$,
\begin{equation}
P_{\mathrm{sh}}(a,b)=O(d^{-N}), \qquad P_{\mathrm{ind}}(a,b)=O(d^{-N}).
\end{equation}
Hence
\begin{equation}
\mathrm{TVD}_{S_{\ge2}}=O(d^{-2}).
\end{equation}
We now analyze configurations containing at most one collision. For
$(a,b)\in S_0$, since $(a_j,b_j)\in\{(1,0),(0,1),(0,0)\}$, all local factors
satisfy
\begin{equation}
\binom{a_j+b_j}{a_j}=1.
\end{equation}
Therefore
\begin{equation}
\frac{P_{\mathrm{sh}}(a,b)}{P_{\mathrm{ind}}(a,b)}=\frac{M_{1}M_{2}}{M\binom{N}{n_1}}.
\end{equation}
Using the asymptotic expansion already employed in Proposition~\ref{fixedN},
\begin{equation}
\begin{split}
    & \binom{k+d-1}{d-1}=\frac{d^{k}}{k!}\left(1+\frac{k(k-1)}{2d}+O(d^{-2})\right).\\
\end{split}
\end{equation}
we obtain
\begin{equation}
\frac{P_{\mathrm{sh}}(a,b)}{P_{\mathrm{ind}}(a,b)}=1-\frac{n_1n_2}{d}+O(d^{-2}).
\end{equation}
Consequently
\begin{equation}
\left|\frac{P_{\mathrm{sh}}(a,b)}{P_{\mathrm{ind}}(a,b)}-1\right|=\frac{n_1n_2}{d}+O(d^{-2}).
\end{equation}
The complement of \(S_0\) consists of configurations containing at least one
collision. Under the independent-block model, there are $\binom{N}{2}$ pairs of
samples, and each pair-collision event has probability of order $d^{-1}$. By the union bound,
the probability of observing at least one collision is therefore
\begin{equation}
O\!\left(\frac{N^2}{d}\right).
\end{equation}
For fixed \(N\), this reduces to \(O(d^{-1})\), hence
\begin{equation}
\mathbb P_{\mathrm{ind}}(S_0)=1-O(d^{-1}).
\end{equation}
Therefore
\begin{equation}
\mathrm{TVD}_{S_0}=\frac12\left(\frac{n_1n_2}{d}+O(d^{-2})\right)\left(1-O(d^{-1})\right) =\frac{n_1n_2}{2d}+O(d^{-2}).
\end{equation}
Now consider $(a,b)\in S_{\mathrm{int}}$. Then there exists a unique coordinate
such that
\begin{equation}
(a_j,b_j)=(2,0)\qquad\text{or}\qquad(a_j,b_j)=(0,2).
\end{equation}
For this coordinate,
\begin{equation}
\binom{2}{2}=\binom{2}{0}=1,
\end{equation}
while all remaining factors are also equal to \(1\). Consequently,
\begin{equation}
\frac{P_{\mathrm{sh}}(a,b)}{P_{\mathrm{ind}}(a,b)}=1-\frac{n_1n_2}{d}+O(d^{-2}),
\end{equation}
exactly as in the collision-free case. Then
\begin{equation}
    \left|\frac{P_{\mathrm{sh}}(a,b)}{P_{\mathrm{ind}}(a,b)}-1\right|=\frac{n_1n_2}{d}+O(d^{-2}).
\end{equation}
This event requires an internal collision within one of the two blocks. In the
first block there are $\binom{n_1}{2}$ pairs of samples, while in the second
block there are $\binom{n_2}{2}$ such pairs. For two samples drawn from the same Haar-distributed probability vector,
the collision probability equals
\begin{equation}
\sum_{j=0}^{d-1}\mathbb E[W_j^2]=\frac{2}{d+1}.
\end{equation}
Therefore
\begin{equation}
\mathbb P_{\mathrm{ind}}(S_{\mathrm{int}})=O\!\left(\frac{n_1^2+n_2^2}{d}\right).
\end{equation}
For fixed \(N=n_1+n_2\), this is \(O(d^{-1})\). Consequently,
\begin{equation}
\mathrm{TVD}_{S_{\mathrm{int}}}=\frac12 \left(\frac{n_1n_2}{d}+O(d^{-2})\right)O(d^{-1}) =O(d^{-2}).
\end{equation}
Thus internal collisions do not contribute at leading order. Finally, consider
$(a,b)\in S_\times$. Then, for the unique collision coordinate,
\begin{equation}
(a_j,b_j)=(1,1),
\end{equation}
and therefore
\begin{equation}
\binom{a_j+b_j}{a_j}=\binom21=2.
\end{equation}
All remaining factors equal $1$, hence
\begin{equation}
\frac{P_{\mathrm{sh}}(a,b)}{P_{\mathrm{ind}}(a,b)}=2\left(1-\frac{n_1n_2}{d}+O(d^{-2})\right)=2-\frac{2n_1n_2}{d}+O(d^{-2}).
\end{equation}
In this case,
\begin{equation}
\left|\frac{P_{\mathrm{sh}}(a,b)}{P_{\mathrm{ind}}(a,b)}-1\right|=\left|1-\frac{2n_1n_2}{d}+O(d^{-2}) \right|=1+O(d^{-1}).
\end{equation}
A configuration in \(S_\times\) contains exactly one cross-block collision.
There are $n_1n_2$ pairs consisting of one sample from the first block and one
sample from the second block. For each such pair, the probability that both
samples occupy the same bin is $\frac{1}{d}$. Hence the expected number of
cross-block collisions equals $\frac{n_1n_2}{d}$. The probability of observing
two or more cross-block collisions is \(O(d^{-2})\), since this requires at least two distinct pairs of samples to collide simultaneously. Indeed, each additional
collision contributes another factor \(d^{-1}\). Therefore
\begin{equation}
\mathbb P_{\mathrm{ind}}(S_\times)=\frac{n_1n_2}{d}+O(d^{-2}).
\end{equation}
Hence
\begin{equation}
\mathrm{TVD}_{S_\times}=\frac12\left(\frac{n_1n_2}{d}+O(d^{-2})\right)\left(1+O(d^{-1})\right)=\frac{n_1n_2}{2d}+O(d^{-2}).
\end{equation}
Collecting all contributions gives
\begin{equation} 
\mathrm{TVD}\left(P_{\mathrm{sh}},P_{\mathrm{ind}}\right)=\frac{n_1n_2}{d}+O(d^{-2}). 
\end{equation}
The same argument extends to the sparse regime \(N=o(\sqrt d)\). Indeed, the
probability of observing at least two collisions is now
\begin{equation}
O\!\left(\frac{N^4}{d^2}\right),
\end{equation}
since specifying two collisions requires choosing two pairs of samples, which
can be done in \(O(N^4)\) ways, while each collision contributes a factor
\(d^{-1}\). Moreover,
\begin{equation}
\mathbb P_{\mathrm{ind}}(S_\times)=\frac{n_1n_2}{d}+O\!\left(\frac{N^4}{d^2}\right),
\end{equation}
while all remaining sectors contribute only
\begin{equation}
O\!\left(\frac{N^4}{d^2}\right).
\end{equation}
Consequently,
\begin{equation}
\mathrm{TVD}\left(P_{\mathrm{sh}},P_{\mathrm{ind}}\right)=\frac{n_1n_2}{d}+O\!\left(\frac{N^4}{d^2}\right).
\end{equation}

\subsection{Proof of Theorem~\ref{CriticalResolved}} \label{ProofCritResolved}

We again start from the likelihood ratio
\begin{equation}
\frac{P_{\mathrm{sh}}(a,b)}{P_{\mathrm{ind}}(a,b)}=\frac{M_{1}M_{2}}{M\binom N{n_1}}\prod_{j=0}^{d-1}\binom{a_j+b_j}{a_j}.
\end{equation}
Recall that in the critical scaling regime
\begin{equation}
\frac{N}{\sqrt d}\to c, \qquad \frac{n_1}{N}\to\alpha, \qquad \frac{n_2}{N}\to 1-\alpha.
\end{equation}
We define the number of cross-block collisions by
\begin{equation}
C(a,b):=\#\{j:(a_j,b_j)=(1,1)\}.
\end{equation}
Thus $C(a,b)$ counts the number of bins simultaneously occupied by exactly one
sample from the first block and exactly one sample from the second block. Since
\begin{equation}
n_1\sim \alpha N, \qquad n_2\sim (1-\alpha)N, \qquad \frac{N}{\sqrt d}\to c,
\end{equation}
we have
\begin{equation}
\frac{n_1n_2}{d}\longrightarrow \alpha(1-\alpha)c^2=: \beta.
\end{equation}
We now compute the factorial moments of $C(a,b)$ under the independent-block
model. For fixed $r\ge1$ and all sufficiently large $d$,
\begin{equation}
\begin{split}
\mathbb E_{\mathrm{ind}}\big[(C)_r\big]
&=(d)_r
\frac{\binom{n_1-r+d-r-1}{d-r-1}}{\binom{n_1+d-1}{d-1}}
\frac{\binom{n_2-r+d-r-1}{d-r-1}}{\binom{n_2+d-1}{d-1}}.
\end{split}
\end{equation}
Indeed, choose the ordered list of $r$ distinct bins with $(a_j,b_j)=(1,1)$ and
then distribute the remaining $n_1-r$ and $n_2-r$ counts over the remaining
$d-r$ bins. Since $n_1,n_2=O(\sqrt d)$, the two ratios are asymptotic to
$(n_1/d)^r$ and $(n_2/d)^r$, respectively. Hence
\begin{equation}
\mathbb E_{\mathrm{ind}}\big[(C)_r\big]\longrightarrow \beta^r.
\end{equation}
The method of factorial moments gives
\begin{equation}
C(a,b)\xrightarrow[]{d}\mathrm{Poisson}(\beta).
\end{equation}
It remains to check that non-simple collision configurations do not affect the
likelihood ratio. Let
\begin{equation}
B_d:=\{a_j+b_j\le2\ \text{for all }j\}.
\end{equation}
By a union bound over triples of samples, using the Dirichlet moment estimates
$\mathbb E\sum_j W_j^3=O(d^{-2})$ and
$\mathbb E\sum_j W_j^2Z_j=O(d^{-2})$ for independent
$W,Z\sim\mathrm{Dirichlet}(1,\dots,1)$, we get
\begin{equation}
\mathbb P_{\mathrm{ind}}(B_d^c)=O\!\left(\frac{N^3}{d^2}\right)=o(1).
\end{equation}
We now analyze the product
\begin{equation}
\prod_{j=0}^{d-1}\binom{a_j+b_j}{a_j}.
\end{equation}
If a bin contains no collision, then
\begin{equation}
(a_j,b_j)\in\{(1,0),(0,1),(0,0)\},
\end{equation}
and therefore
\begin{equation}
\binom{a_j+b_j}{a_j}=1.
\end{equation}
If a bin contains an internal collision, i.e.
\begin{equation}
(a_j,b_j)=(2,0) \qquad\text{or}\qquad (a_j,b_j)=(0,2),
\end{equation}
then again
\begin{equation}
\binom22=\binom20=1.
\end{equation}
Hence internal collisions do not affect the likelihood ratio. The only
nontrivial contribution comes from cross-block collisions:
\begin{equation}
(a_j,b_j)=(1,1),
\end{equation}
for which
\begin{equation}
\binom{a_j+b_j}{a_j}=\binom21=2.
\end{equation}
Consequently, every cross-block collision contributes a factor \(2\), and
therefore on $B_d$
\begin{equation}
\prod_{j=0}^{d-1}\binom{a_j+b_j}{a_j}=2^{C(a,b)}.
\end{equation}
Recall
\begin{equation}
\overline M=\frac{M_{1}M_{2}}{M}.
\end{equation}
By Lemma~\ref{normalizationlem},
\begin{equation}
\frac{\binom N{n_1}}{\overline M}\longrightarrow e^\beta.
\end{equation}
Equivalently,
\begin{equation}
\frac{\overline M}{\binom N{n_1}}=\frac{M_{1}M_{2}}{M\binom N{n_1}}\longrightarrow e^{-\beta}.
\end{equation}
Combining this with the previous step yields
\begin{equation}
\frac{P_{\mathrm{sh}}(a,b)}{P_{\mathrm{ind}}(a,b)}=\frac{M_{1}M_{2}}{M\binom N{n_1}}\prod_{j=0}^{d-1}\binom{a_j+b_j}{a_j}\Longrightarrow e^{-\beta}2^C,
\end{equation}
where
\begin{equation}
C\sim \mathrm{Poisson}(\beta).
\end{equation}
Let
\begin{equation}
R_d(a,b):=\frac{P_{\mathrm{sh}}(a,b)}{P_{\mathrm{ind}}(a,b)}.
\end{equation}
The variables $R_d$ are nonnegative and satisfy
\(\mathbb E_{\mathrm{ind}}[R_d]=1\) exactly. The limiting variable also has mean
one, since
\begin{equation}
\mathbb E\!\left[e^{-\beta}2^C\right]
=e^{-\beta}\exp(\beta(2-1))=1.
\end{equation}
Thus convergence in distribution together with convergence of the first moments
implies uniform integrability of $\{R_d\}$, and hence of $\{|R_d-1|\}$. Using the
likelihood-ratio representation
\begin{equation}
\mathrm{TVD}(P_{\mathrm{sh}},P_{\mathrm{ind}})=\frac12\mathbb E_{\mathrm{ind}}\left[\left|R_d(a,b)-1\right|\right],
\end{equation}
we conclude that
\begin{equation}
\mathrm{TVD}(P_{\mathrm{sh}},P_{\mathrm{ind}})\longrightarrow\frac12\mathbb E\left[\left|e^{ -\beta}2^C-1\right| \right].
\end{equation}
This proves the claim.

\subsection{Proof of~\ref{ResolvedFixedd}} \label{ProofResolvedFixedd}
Fix \(d\ge2\) and let \(N\to\infty\) with $\frac{n_1}{N}\to\alpha\in(0,1)$. Under
the shared-unitary model, conditioned on
\begin{equation}
W\sim\mathrm{Unif}(\Delta_{d-1}),
\end{equation}
all samples in both blocks are i.i.d. with distribution \(W\). Therefore, by the
law of large numbers,
\begin{equation}
\left(\frac{a}{n_1},\frac{b}{n_2}\right)\xrightarrow[]{p}(W,W).
\end{equation}
Under the independent-block model, the two blocks are generated independently
from independent Haar-random probability vectors \(W_1,W_2\). Hence
\begin{equation}
\left(\frac{a}{n_1},\frac{b}{n_2}\right)\xrightarrow[]{p}(W_1,W_2),
\end{equation}
where \(W_1,W_2\) are independent and uniformly distributed on \(\Delta_{d-1}\).
Thus the shared-model limit is supported on the diagonal
\begin{equation}
\mathcal D=\{(w,w):w\in\Delta_{d-1}\},
\end{equation}
whereas the independent-model limit is absolutely continuous with respect to
Lebesgue measure on
\begin{equation}
\Delta_{d-1}\times\Delta_{d-1}.
\end{equation}
Since the diagonal \(\mathcal D\) has codimension \(d-1\), it has Lebesgue
measure zero. Therefore
\begin{equation}
\mathbb P\big((W_1,W_2)\in\mathcal D\big)=0.
\end{equation}
Hence the two limiting measures are mutually singular. To turn this into a total
variation statement, for $\delta>0$ let
\begin{equation}
\mathcal D_\delta:=\{(x,y)\in\Delta_{d-1}\times\Delta_{d-1}:\|x-y\|_1\le\delta\}.
\end{equation}
Under the shared model,
\begin{equation}
\mathbb P_{\mathrm{sh}}\left(\left(\frac a{n_1},\frac b{n_2}\right)\in\mathcal D_\delta\right)\longrightarrow 1
\end{equation}
for every fixed $\delta>0$. Under the independent model, weak convergence gives
\begin{equation}
\limsup_{N\to\infty}
\mathbb P_{\mathrm{ind}}\left(\left(\frac a{n_1},\frac b{n_2}\right)\in\mathcal D_\delta\right)
\le
\mathbb P\big(\|W_1-W_2\|_1\le\delta\big).
\end{equation}
Because $W_1$ and $W_2$ have a joint density and $d\ge2$,
$\mathbb P(\|W_1-W_2\|_1\le\delta)\to0$ as $\delta\downarrow0$. Therefore, for
arbitrarily small $\varepsilon>0$ we may choose $\delta$ so that the independent
model assigns asymptotic mass at most $\varepsilon$ to $\mathcal D_\delta$, while
the shared model assigns asymptotic mass one. Consequently,
\begin{equation}
\liminf_{N\to\infty}\mathrm{TVD}(P_{\mathrm{sh}},P_{\mathrm{ind}})\ge1-\varepsilon.
\end{equation}
Since total variation distance is always at most one and $\varepsilon$ is
arbitrary, the desired convergence follows.
\end{document}